\documentclass[12pt]{article}
\usepackage{euscript,amsmath, amssymb, amsfonts}
\usepackage{color}

\newcommand{\R}{{\mathbb R}}

\newcommand{ \bee}{\begin{eqnarray}}
\newcommand{ \eee}{\end{eqnarray}}

\newcommand{\p}{{\hbar}}

\newcommand{\oC}{{\mathbb C}}

\newcommand{\di}{{ d }}
\newcommand{\D}{{\cal D}}
\newcommand{\eE}{{\cal E}}

\newcommand{\G}{{\mathbb G}}
\newcommand{\Z}{{\mathbb Z}}
\newcommand{\K}{{\mathbb K}}
\newcommand{\be}{\begin{equation}}
\newcommand{\ee}{\end{equation}}

\newcounter{theorem}
\newcommand{\theorem}{\par\refstepcounter{theorem}
           {\bf Theorem \arabic{section}.\arabic{theorem}. }}
\renewcommand\thetheorem{\thesection.\arabic{theorem}}
\makeatletter \@addtoreset{theorem}{section}

\newcounter{lemma}
\newcommand{\lemma}{\par\refstepcounter{lemma}
           {\bf Lemma \arabic{section}.\arabic{lemma}. }}
\renewcommand\thelemma{\thesection.\arabic{lemma}}
\makeatletter \@addtoreset{lemma}{section}

\newcounter{proposition}
\newcommand{\proposition}{\par\refstepcounter{proposition}
           {\bf Proposition \arabic{section}.\arabic{proposition}. }}
\renewcommand\theproposition{\thesection.\arabic{proposition}}
\makeatletter \@addtoreset{proposition}{section}

\makeatletter \@addtoreset{equation}{section}
\def\theequation{\thesection.\arabic{equation}}
\makeatother

\newcounter{appen}
\newcommand{\appen}[1]{\par\refstepcounter{appen}
{\par\medskip\noindent\Large\bf Appendix \arabic{appen}.
\medskip
}{\Large\bf #1}}
\newenvironment{proof}[1][Proof]{\noindent\textsf{#1.\ }}{ \ \rule{0.5em}{0.5em}}

\newcounter{subappen}
\newcommand{\subappen}[1]{\par\refstepcounter{subappen}
{\par\vskip 1cm\noindent\large\bf  \arabic{appen}.\arabic{subappen}.
}{\large\bf #1}\vskip 0.5 cm}
\makeatletter \@addtoreset{subappen}{appen}

\newcounter{subsubappen}
\newcommand{\subsubappen}[1]{\par\refstepcounter{subsubappen}
{\par\medskip\noindent\bf \arabic{appen}.\arabic{subappen}.\arabic{subsubappen}.
\medskip
}{\bf #1}}
\makeatletter \@addtoreset{subsubappen}{subappen}

\font\frtnfr=eufm10   scaled\magstep1
\font\twlfr=eufm10
\font\tenfr=eufm10

\newfam\frfam
\textfont\frfam=\frtnfr
\scriptfont\frfam=\twlfr
\scriptscriptfont\frfam=\tenfr

\font\frtnopen=msbm10  scaled\magstep2
\font\twlopen=msbm10
\font\tenopen=msbm10

\newfam\openfam
\textfont\openfam=\frtnopen
\scriptfont\openfam=\twlopen
\scriptscriptfont\openfam=\tenopen

\font\frtnsf = cmss12 scaled\magstep1
\font\twlsf = cmss10
\font\tensf = cmss9

\newfam\Scfam
\textfont\Scfam = \frtnsf
\scriptfont\Scfam = \twlsf
\scriptscriptfont\Scfam = \tensf

\begin{document}


\sloppy \title
 {
The deformations of antibracket
with even and odd deformation parameters, defined on the space $DE_1$}
\author
 {
 S.E.Konstein\thanks{E-mail: konstein@lpi.ru}\ \ and
 I.V.Tyutin\thanks{E-mail: tyutin@lpi.ru}
 \thanks{
               This work was supported
               by the RFBR (grant No.~11-02-00685)
 } \\
               {\sf \small I.E.Tamm Department of
               Theoretical Physics,} \\ {\sf \small P. N. Lebedev Physical
               Institute,} \\ {\sf \small 119991, Leninsky Prospect 53,
               Moscow, Russia.} }
\date{
}

\maketitle

\maketitle

\begin{abstract}
{ \footnotesize We consider antiPoisson superalgebra realized on the smooth
Grassmann-valued functions of the form $\xi f_0(x)+f_1(x)$,
where $f_0$ has compact support on $\R^1$,
 and with the parity opposite to
that of the Grassmann superalgebra realized on these functions.
The deformations with even and odd deformation parameters
 of this superalgebra
are found.
}
\end{abstract}

\section{Introduction}

Let $\G^n$ be a Grassmann algebra on $n$ indeterminates.
In \cite{apoi} we described the deformations of antiPoisson superalgebra
realized on the space $\mathbf D_{n}$ of smooth
$\G^n$-valued functions with compact support on $\R^n$ and show that there exists either
one deformation with one even deformation parameter, or one deformation with one odd parameter.
During the proof of this statement, the second cohomology space $H^2(\mathbf D_{n}, \mathbf E_{n})$
was calculated, where $\mathbf E_{n}$ is the space of smooth
$\G^n$-valued functions on $\R^n$,
 and it was shown that $dim H^2(\mathbf D_{n}, \mathbf E_{n})=1|1$ if $n\ge 2$
and $dim\,H^2(\mathbf D_{1}, \mathbf E_{1})=3|3$.

Let $\mathbf {DE}$ be the space of $\G^1$-valued functions of the
form $f(x,\xi)=\xi f_0(x)+f_1(x)$, where $f_0$ and $f_1$ are smooth
functions on $\R^1$, such that $f_0$ has compact support, and $\xi$
is the only generating element of the Grassmann algebra $\G^1$.

Here we explore the observation, that
$H^2(\mathbf D_{1}, \mathbf E_{1})=H^2(\mathbf {D}_{1}, \mathbf {DE})\subset
H^2( \mathbf {DE}, \mathbf  {DE})$, and that $ \mathbf {DE}$ may have a deformation
with several even and several odd deformation parameters.

In the present work, we found the deformations of Poisson antibracket
realized on $ \mathbf {DE}$.

Particularly, we found all the deformations with one even and at most three odd deformation parameters,
and all the deformation of some particular form with arbitrary number of odd parameters.

The text is organized as follows.
For background, see Section \ref{nota}.
Section \ref{cohom-res} contains calculation of
$H^2(\mathbf{DE},\mathbf{DE})$.
This space is much wider than $H^2(\mathbf D_{1}, \mathbf E_{1})$ and
the additional cocycles are parameterized by the elements of the quotient space
${\cal D}' / ({C^\infty})'$

Theorems \ref{-th1} and \ref{-th2} describing the deformations are formulated in Section \ref{def-res}.

Main Theorems (\ref{cc} and \ref{-th2}) are proved in Appendices \ref{-app1}-\ref{app6}.

\section{Preliminary and notation}\label{nota}

Let $\K$ be either $\R$ or $\oC$. We denote by ${\cal D}(\R^n)$ the space
of smooth $\K$-valued functions with compact supports on $\R^n$. This space
is endowed with its standard topology. We set
$$
 \mathbf D_{n}= {\cal
D}(\R^{n})\otimes \G^{n},\quad \mathbf E_{n}=
C^\infty(\R^{n})\otimes \G^{n},
$$
The generators of
the Grassmann algebra (resp., the coordinates of the space $\R^{n}$) are
denoted by $\xi^\alpha$, $\alpha=1,\ldots,n$ (resp., $x^i$, $i=1,\ldots,
n$).
We also use collective variables $z^A$ which are equal to $x^A$
for $A=1,\ldots,n$ and are equal to $\xi^{A-n}$ for
$A=n+1,\ldots,2n$.

The spaces $\mathbf D_{n}$ and $\mathbf
E_{n}$  possess a natural parity
which is determined by that of the Grassmann algebra.
The Grassmann parity  of an
element $f$ of these spaces is denoted by $\varepsilon(f)$.

The spaces $\mathbf D_{n}$ and $\mathbf
E_{n}$ possess also another parity
$\epsilon$
 which is opposite to the $\varepsilon$-parity: $\epsilon=\varepsilon+1$.

We set
$\varepsilon_A=0$, $\epsilon_A=1$ for $A=1,\ldots, n$
and $\varepsilon_A=1$, $\epsilon_A=0$ for
$A=n+1,\ldots, 2n$.

It is well known, that the bracket
\bee\label{Sch}
[f,g](z)=\sum_{i=1}^n\left(f(z)\frac{\overleftarrow{\partial}}{\partial x^i}
\frac\partial{\partial\xi^i}g(z)-
f(z)\frac{\overleftarrow{\partial}}{\partial\xi^i}
\frac{\partial}{\partial x^i}g(z)\right),
\eee
which we call ''antibracket'', %
defines the structure of Lie superalgebra on the superspaces
$\mathbf D_{n}$ and
$\mathbf E_{n}$
with the $\epsilon$-parity.


We call these Lie superalgebras \emph{antiPoisson
superalgebras}.\footnote{
We will also consider associative multiplication of the elements of considered antiPoisson
superalgebras with commutation relations $fg=(-1)^{\varepsilon(f)\varepsilon(g)}gf$.}

The integral on $\mathbf D_{n}$ is defined by the relation $ \int \di z\,
f(z)= \int_{\R^{n}}\di x\int \di\xi\, f(z), $ where the integral on the
Grassmann algebra is normed by the condition $\int \di\xi\,
\xi^1\ldots\xi^{n}=1$. We identify $\G^{n}$ with its dual space $\G^{\prime
n}$ setting $f(g)=\int\di\xi\, f(\xi)g(\xi)$, $f,g\in \G^{n}$.
Correspondingly, the space $\mathbf D^{\prime}_{n}$ of continuous linear
functionals on $\mathbf D_{n}$ is identified  with the space
$\D^\prime(\R^{n})\otimes \G^{n}$
and $\mathbf E^{\prime}_{n}$
is identified  with the space
$(C^\infty(\R^{n}))^\prime\otimes \G^{n}$.

The value $m(f)$ of a functional
$m\in\mathbf D^\prime_n$ on a function $f\in\mathbf D_n$ will be often
written in the integral form: $m(f)=\int\di z\,m(z)f(z)$.

Introduce the superalgebra $\mathbf {DE}_{n}$, $\mathbf D_{n}\subset \mathbf {DE}_{n} \subset\mathbf E_{n}$
by the relation
\be\label{den}
\mathbf {DE}_{n}=\{f\in\mathbf {E}_{n}:\ f-\int \di\xi\,
\xi^1\ldots\xi^{n} f \in  \mathbf D_{n} \}.
\ee
Clearly, if $f\in \mathbf {DE}_{n}$ and $g\in\mathbf {DE}_{n}$, then $[f,g]\in \mathbf {D}_{n}$.

Below we consider the case $n=1$ only. For simplicity we will denote
${\cal
D}(\R)$ as $D$,
$C^\infty(\R)$ as $E$,
$\mathbf {DE}_{1}$ as $ \mathbf {DE}$.
Besides we will denote $\mathbf {D}_{1}$ as $\mathbf {D}$, $\mathbf {E}_{1}$ as $\mathbf {E}$.

It follows from (\ref{den}) that $ \mathbf {DE}$ consists of the functions of the form $f=\xi f_0(x)+f_1(x)$
where $f_0\in D$ and $f_1\in E$.


\section{Cohomology of antibrackets }\label{cohom-res}

Let $\mathbf U$ be either $\mathbf D$ or $\mathbf {DE}$.
Let $\mathbf U   $ acts in a $\Z_2$-graded space $V$ (the action of $f\in
\mathbf U   $ on $v\in V$ will be denoted by $f\cdot v$). The space
$C_p(\mathbf U   ,V)$ of $p$-cochains consists of all multilinear
superantisymmetric mappings from $\mathbf U   ^p$ to $V$. Superantisymmetry
means, as usual, that $M_p(\,...\,,f_i,\,f_{i+1},...)=-
(-1)^{\epsilon(f_i)\epsilon(f_{i+1})}M_p(\,...,f_{i+1},f_{i},...)$. The space
$C_p(\mathbf U   , V)$ possesses a natural $\Z_2$-parity: by definition,
$M_p\in C_p(\mathbf U   ,V)$ has the definite parity
$\epsilon_{M_p}$ if
$$
\epsilon(M_p(f_1,\ldots,f_p))=
\epsilon_{M_p}+\epsilon(f_1)+\ldots+\epsilon(f_p)
$$
for any $f_j\in\mathbf U   $ with $\epsilon$-parities $\epsilon(f_j)$.
Sometimes we
will use the Grassmann $\varepsilon$-parity\footnote{If $V$ is the
space of Grassmann-valued functions on $\R^n$ then $\varepsilon$ defined in
such a way coincides with usual Grassmann parity. } of cochains:
$\varepsilon_{M_p}=\epsilon_{M_p}+p+1$.
The differential $\di_p^V$ is defined to be
the linear operator from $C_p(\mathbf U   , V)$ to $C_{p+1}(\mathbf U   , V)$
such that
\begin{eqnarray}
&&d_p^{V}M_p(f_1,...,f_{p+1})=
-\sum_{j=1}^{p+1}(-1)^{j+\epsilon(f_j)|\epsilon(f)|_{1,j-1}+
\epsilon(f_j)\epsilon_{M_p}}f_j\cdot
M_p(f_{1},...,\breve{f}_j,...,f_{p+1})- \nonumber \\
&&-\sum_{i<j}(-1)^{j+\epsilon(f_j)|\epsilon(f)|_{i+1,j-1}}
M_p(f_1,...f_{i-1},[f_i,f_j],f_{i+1},...,
\breve{f}_j,...,f_{p+1}),\label{diff}
\end{eqnarray}
for any $M_p\in C_p(\mathbf U   , V)$ and $f_1,\ldots,f_{p+1}\in\mathbf
U   $ having definite $\epsilon$-parities. Here the sign $\breve{}\ $  means
that the argument is omitted and the notation
$$
|\epsilon(f)|_{i,j}=\sum_{l=i}^j\epsilon(f_l)
$$
 has been used.
 The
differential $\di^V$ is nilpotent (see \cite{Schei97}), {\it i.e.},
$\di^V_{p+1}\di^V_p=0$ for any $p=0,1,\ldots$.
The $p$-th cohomology space of
the differential $\di_p^V$ will be denoted by $H^p_V$. The second cohomology
space $H^2_{\mathrm{ad}}$ in the adjoint representation is closely related to
the problem of finding formal deformations of the Lie bracket $[\cdot,\cdot]$
of the form $ [f,g]_*=[f,g]+\hbar[f,g]_1+\ldots$ up to
similarity transformations $[f,g]_T=T^{-1}[Tf,Tg]$ where continuous
linear operator $T$
from $V[[\p]]$ to $V[[\p]]$ has the form
$T=\mathsf {id}+\p T_1$.

The condition that
$[\cdot,\cdot]_1$ is a 2-cocycle is equivalent to the Jacobi identity for
$[\cdot,\cdot]_*$ modulo the $\hbar$-order terms.

In the present paper, similarly to \cite{SKT1}, we suppose that cochains are
separately continuous multilinear mappings.

We need the cohomologies of the antiPoisson algebra
$ \mathbf {DE}$ in the
adjoint representation:
 $V=\mathbf {DE}$ and $f\cdot g=[f,g]$ for any
$f,g\in\mathbf {DE}$. The space $C_p(\mathbf  {DE}, \mathbf  {DE})$ consists
of separately continuous superantisymmetric multilinear mappings from $(\mathbf
 {DE})^p$ to $\mathbf  {DE}$.  In fact, we need here the case $p=2$ only.

We call
$p$-cocycles $M_p^1,\ldots M_p^k$ independent cohomologies if they give rise to linearly
independent elements in $H^p$. For a multilinear form $M_p$ taking values in
$\mathbf D$, $\mathbf E$, or $\mathbf  {DE}$, we write
$M_p(z|f_1,\ldots,f_p)$ instead of more cumbersome $M_p(f_1,\ldots,f_p)(z)$.

The following theorem proved in \cite{JMP} and \cite{apoi} describes $p=2$ cohomology of antibracket

\theorem\label{th2}{}
{\it
Let the bilinear mappings $m_{2|1}$, $m_{2|2}$, $m_{2|3}$, $m_{2|4}$, $m_{2|5}$,
and $m_{2|6}$
from $(\mathbf D)^2$ to
$\mathbf E$ be defined by the relations
\begin{eqnarray}
&&m_{2|1}(z|f,g)=\int du \partial_\eta g(u)\partial^{3}_y f(u),\;\;
\epsilon_{m_{2|1}}=1, \label{5.2.7} \\
&&m_{2|2}(z|f,g)=\int du\theta(x-y)[\partial_\eta g(u)\partial_y^3f(u)-
\partial_\eta f(u)\partial_y^3 g(u)]+ \nonumber \\
&&+x[\{\partial_\xi\partial_{x}^{2}f(z)\}\partial_\xi\partial_{x}g(z)-
\{\partial_\xi\partial_xf(z)\}\partial_\xi
\partial_x^2g(z)],\;\; \epsilon_{m_{2|2}}=1,
\label{5.2.8} \\
&&m_{2|3}(z|f,g)=(-1)^{\varepsilon(f)}\{(1-N_{\xi})f(z)\}(1-N_{\xi})g(z), \;\;
\epsilon_{m_{2|3}}=1, \label{6.3a} \\
&&m_{2|4}(z|f,g)=(-1)^{\varepsilon (f)}\{\Delta f(z)\}\eE_z g(z)+
\{\eE_z f(z)\}\Delta g(z)\;\;  \epsilon_{m_{2|4}}=0. \label{6.3b}\\
&&m_{2|5}(z|f,g)=\int du (-1)^{\epsilon(f)}\partial_y f(u)\partial_y g(u),\;\;
\epsilon_{m_{2|5}}=0, \label{5.2.7-5} \\
&&m_{2|6}(z|f,g)=\int du\theta(x-y)(-1)^{\epsilon(f)}\partial_y f(u)\partial_y g(u),\;\;
\epsilon_{m_{2|6}}=0 \label{5.2.7-6}
\end{eqnarray}
where $z=(x,\xi)$, $u=(y,\eta)$, $N_\xi=\xi \partial_{\xi}$, and
\begin{equation}
\Delta =\partial_x \partial_\xi, \quad
\mathcal{E}_{z}=1-\frac 1 2(x\partial_x +\xi\partial_\xi)
\label{Delta}
\end{equation}

Then
$dim H^2({\mathbf D,\mathbf E})=3|3$ and the cochains
$m_{2|1}(z|f,g)$, $m_{2|2}(z|f,g)$, $m_{2|3}(z|f,g)$, $m_{2|4}(z|f,g)$, $m_{2|5}(z|f,g)$, and $m_{2|6}(z|f,g)$
are
independent nontrivial cocycles.
}

Observe that in fact all the forms $m_{2|i}$ ($i=1,...,6$)
take the value in $\mathbf {DE}$ and can be extended
from  $\mathbf D^2$ to $\mathbf {DE}^2$ taking the value in $\mathbf {DE}$.
So,
$H^2(\mathbf D, \mathbf E)=H^2(\mathbf {D}, \mathbf {DE})\subset
H^2(\mathbf {DE}, \mathbf {DE})$.

To find
$H^2(\mathbf {DE}, \mathbf {DE})$ we have to
determine whether $\mathbf {DE}$ has 2-cocycles different from
$m_{2|i}$ ($i=1,...,6$).

The following theorem answers this question.

\begin{theorem}\label{m7}
{\it
\begin{enumerate}
\item
Let $M\in D'$, $M\notin E'$ and let $M(f)=\int dx \mu(x)f(x)$ with
some distribution $\mu$.
Then bilinear form $m_{2|7}(M)$ defined as
\be
m_{2|7}(M|f,g)=\int dz \xi \mu(x)[f,g]
\ee
or, equivalently,
\be
m_{2|7}(M|f,g)=M\left(\int d\xi\,\, \xi [f,g]\right)
\ee
is representative of nontrivial element of $H^2(\mathbf {DE},\mathbf {DE})$.
\item
Up to exact form, each element in $H^2(\mathbf {DE},\mathbf {DE})$
has the form
\be\label{coc}
m_2=\sum_{i=1}^6 c_i m_{2|i}+m_{2|7}(M),
\ee
where $m_{2|i}$ are listed in (\ref{5.2.7})-(\ref{5.2.7-6}),
and
$M\in
\!\!\!
{\raise2pt\hbox{
$D'$
}}
\!\!\!\!
\big/
\!\!\!
{\raise-2pt\hbox{
$E'$
}}$.

\end{enumerate}
}\end{theorem}

Note, that if $M\in E'$, then $m_{2|7}(M)=d_1 M$.

\begin{proof}

Let $m_2$ be nontrivial cocycle in ${\cal C}^2(\mathbf {DE}, \mathbf {DE})$,
i.e. $d_2 m_2=0$ and there does  not exist such $m_1\in \cal C^1(\mathbf {DE}, \mathbf {DE})$
that $m_2=d_1m_1$.

Consider the restriction $m_2|_{\mathbf {D}^2}$ of $m_2$ to $\mathbf {D}^2$.
As this restriction is a cocycle in ${\cal C}^2(\mathbf {D}, \mathbf {DE})$,
this restriction has the form
$$
m_2|_{\mathbf {D}^2}=\sum_1^6 c_i m_{2|i} +d_1m_1
$$
with some $m_1\in {\cal C}^1(\mathbf {D}, \mathbf {DE})$.
Denoting $(m_2-\sum_1^6 c_i m_{2|i})$ as $n_2$ we obtain
$$
n_2|_{\mathbf {D}^2}=d_1m_1.
$$

Let us look for such $m_1\in {\cal C}^1(\mathbf {D}, \mathbf {DE})$
that $d_1 m_1$ can be extended to ${\cal C}^2(\mathbf {DE}, \mathbf {DE})$
and $m_1$ can not be extended to ${\cal C}^1(\mathbf {DE}, \mathbf {DE})$.

Straightforward calculation gives that
$m_1(f)=M(\int d\xi\, \xi f)$, where $M\in D'$, $M\notin E'$.

Indeed, $d_1m_1(f,g)=M(\int d\xi\, \xi [f,g])\in {\cal C}^2(\mathbf {DE}, \mathbf {DE})$
because $[f,g]\in \mathbf D$ if $f,g\in \mathbf {DE}$.
\end{proof}

In Appendices \ref{-app1} - \ref{app6} we use
the notation $\omega$ both for the linear functional $M\in D'/E'$ defining
cocycle $m_{2|7}$ in Theorem \ref{m7} and for the kernel of this functional.

\section{Deformation with even and odd parameters. Preliminary.}\label{def-prel}

It occurred, that there exist odd second cohomologies with
coefficients in adjoint representation. It is natural to look for
the deformations associated with these odd cohomologies and having
the odd deformation parameter \cite{Leites}, \cite{leit}.

Below we consider the case, where the functions and multilinear forms may
depend on outer odd parameters $\theta_i$, where $\theta$-s belong to some supercommutative
associative superalgebra ${\cal A}$.
Thus we consider a colored algebras
$\mathbf {DE}\otimes {\cal A}$,  ${\cal C}^p(\mathbf {DE},\mathbf {DE})\otimes {\cal A}$,...
with $(\Z_2)^2$ grading, namely
the grading of element $\theta\otimes f$ is $(\varepsilon_1 (\theta),
\epsilon_2 (f))$.

We preserve the notation $\mathbf {DE}$ for $\mathbf {DE}\otimes {\cal A}$.

We can consider $\mathbf {DE}$ as a Lie superalgebra with the parity $\epsilon=\varepsilon_1
+\epsilon_2$. One can easily check that such consideration is selfconsistent
(see also \cite{Sche1} and discussion on Necludova and Sheunert theorems in \cite{Sosrus1}.

As the second cohomology space with
coefficients in adjoint representation has the dimension $3|3$, it is natural to suppose
that ${\cal A}=\G^3$. For this case all the deformations are found
(see Theorem \ref{-th1} and Theorem \ref{-th2}) and listed in Subsection \ref{listdef}.

Nevertheless, it is possible to consider other superalgebras ${\cal A}$.
In what follows, we consider the deformation of the form
\be
C=\sum_{l=0}^\infty \hbar^l C_l(\theta)
\ee
where $\sum_{l=0}^k \hbar^l C_l(\theta)\in {\cal C}^2(\mathbf {DE},\mathbf {DE})\otimes {\G^{s_k}}$
and $s_0\le s_1\le s_2 \, ...\,$.

During the proofs, we assume that all the $\theta$-s has the same order that $\hbar$,
and sometimes we stress this by writing $C(\hbar, \hbar\theta_i)$.

\section{Algebra of Jacobiators}\label{jac-res}

Before formulate (and prove) the main theorem of present work, which describes
the deformations of antibracket on $\mathbf {DE}$, we have to calculate
all the Jacobiators $J_{i,j}=J(m_{2|i},m_{2|j})$ for $i,j=1,...,7)$.

As it was mentioned above, we define Jacobiator of bilinear forms by nonlinear way.
Let $m$ and $n$ be bilinear forms. Then
\bee
J(m, n)(f,g,h)& = &(-1)^{\epsilon(f)\epsilon(h)}
\left(m(n(f,g),h)+(-1)^{\epsilon(m)\epsilon (n)}n(m(f,g),h)\right)+\notag\\
&& +cycle(f,g,h) \quad \mbox{if $m\ne n$}\\
J(m)(f,g,h)& = &\left\{
{\begin{array}{l}
 {\frac 1  2} J(m, m)(f,g,h) \quad\mbox{if $\epsilon(m)=0$,}\\
 0 \quad\mbox{if $\epsilon(m)=1$.}
\end{array}}
\right.
\eee

Jacobiator is connected with the differential by the relation
\be
d_2 m(f,g,h)=-(-1)^{\epsilon(f)\epsilon(h)}J(m_{2|0}, m)(f,g,h),
\ee
where $m_{2|0}(f,g)=[f,\,g]$

The Jacobiators of each pair of cocycles $J_{i,j}=J(m_{2|i}, m_{2|j})$ are described here
for $i,j=1,...,7$

In the next table stars denote all the pairs $i,j$ such that $J_{i,j}\ne 0$
and has not the form $J(m_{2|0}, n_2)$ with some 2-form $n$.
The sign "X" denote all the pairs $i,j$ such that we don't need $J_{i,j}$
for deriving the deformations of $\mathbf {DE}$:

\vskip 3mm

\begin{tabular}{r ccccccc}
$i\setminus j$ \vline & 1 & 2 & 3 & 4 & 5 & 6 & 7\\
\hline
1 \vline  & 0 & 0 & * & * & 0 & 0 & 0\\
2 \vline  & 0 & 0 & * & * & 0 & $J(m_{2|0}, m_{2|9})$ & $J(m_{2|0}, m_{2|10}(M))$\\
3 \vline  & * & * & 0 & * & * & * & *\\
4 \vline  & * & * & * & X & * & * & *\\
5 \vline  & 0 & 0 & * & * & 0 & 0 & 0\\
6 \vline  & 0 & $J(m_{2|0}, m_{2|9})$ & * & * & 0 & 0 & $J(m_{2|0}, m_{2|8}(M))$\\
7 \vline  & 0 & $J(m_{2|0}, m_{2|10}(M))$ & * & * & 0 & $J(m_{2|0}, m_{2|8}(M))$ & 0
\end{tabular}
\vskip 3mm

Expressions for needed nonzero $J_{i,j}$ and for the forms $m_{2|8}$, $m_{2|9}$, $m_{2|10}$ are listed below.

We use in what follows the notation $\widetilde{\theta }$ for
 smooth function such that $\widetilde{\theta }(x)-{\theta }(x)$
has compact support.

\subsubsection{$J_{4,4}$ }
We have no needs to calculate $J_{4,4}$ because we know from \cite{Leites} and \cite{JMP}
that $m_{2|4}$ generates the deformations

\begin{equation*}
{\ [f,\,g]}_{\ast }{\ \!=\![f,g]+(-1)}^{\varepsilon (f)}{\ \{}%
\frac{ c_4}{1+ c_4 N_{z}/2}{\ \Delta f\}}\mathcal{E}_{z}{\ g+\{}%
\mathcal{E}_{z}{\ f\}}\frac{ c_4}{1+ c_4 N_{z}/2}{\ \Delta g}
\end{equation*}%
where $c_4$ is a series in $\hbar$ and $\hbar^2\theta_i\theta_j$ such that $c_4|_{\hbar=0}=0$.

\subsubsection{$J_{4,5}$ }
\begin{eqnarray*}
{\ J}_{4,5} &=&{\ }{\partial }_{x}g_{0}(x)\int dy{\ \partial }_{y}h_{0}{\
\partial }_{y}f_{1}-{\ \partial }_{x}h_{0}(x)\int dy{\ \partial }_{y}g_{0}{\
\partial }_{y}f_{1}+ \\
&&+\int dyf_{1}({\ \partial }_{y}g_{0}{\ \partial }_{y}^{2}h_{0}-{\ \partial
}_{y}h_{0}{\ \partial }_{y}^{2}g_{0})+ \\
&&+\frac{1}{2}\int dyf_{1}({\ \partial }_{y}^{3}g_{0}(1-y{\ \partial }%
_{y})h_{0}-{\ \partial }_{y}^{3}h_{0}(1-y{\ \partial }_{y})g_{0})+ \\
&&+\text{cycles}
\end{eqnarray*}

\subsubsection{$J_{4,6}$ }

\begin{eqnarray*}
J_{46}=&&\!\!\!\!\!\!\!\!\!\! -\frac{1}{2}\int du\theta (x-y)\eta y(-\partial _{y}^{2}\Delta g(u)\Delta
f(u)+\partial _{y}^{2}\Delta f(u)\Delta g(u))h+ \\
&&\!\!\!\!\!\!\!\!\!\!
+\frac{1}{2}x(-\partial _{x}\Delta g(z)\Delta f(z)+\partial _{x}\Delta
f(z)\Delta g(z))(1-\xi \partial _{\xi })h(z)- \\
&&\!\!\!\!\!\!\!\!\!\!
-\Delta f(z)\int du\theta (x-y)\eta \partial _{\eta }\partial
_{y}^{2}g(u)h(u)+\Delta g(z)\int du\theta (x-y)\eta h(u)\partial _{\eta
}\partial _{y}^{2}f(u)+cycles
\end{eqnarray*}

\subsubsection{$J_{4,7}$ }
\begin{eqnarray*}
J_{47} &=&\int du\eta \mu (y)[(\Delta f\partial _{\eta }g-\partial _{\eta
}f\Delta g)(-\partial _{y}+\frac{1}{2}y\partial _{y}^{2})h+(\partial _{\eta
}f\partial _{y}\Delta g-\partial _{y}\Delta f\partial _{\eta }g)(1-\frac{1}{2%
}y\partial _{y})h] \\
&&+{\Delta f(z)}\int du\eta \mu (y)(\partial _{\eta }g\partial
_{y}h-\partial _{\gamma }g\partial _{\eta }h)+cycles
\end{eqnarray*}





\subsubsection{$J_{6,7}$ }
$J_{67} =J(m_{0},m_{8})$, where

\begin{eqnarray}
m_{8} =
\int du
\widetilde \mu(y)
(-1)^{\epsilon (f)}\partial _{y}f(u)\partial _{y}g(u)
\label{m8}
\end{eqnarray}
and
\begin{eqnarray}
\widetilde \mu(x)= \int dy \mu (y)\left(
\theta
(y-x)-\widetilde{\theta }(y)
 \right)
\label{mu}
\end{eqnarray}



\subsubsection{$J_{1,3}$ }
\begin{eqnarray*}
J_{13}
&=&-(1-N_{\xi })h(z)\int du\partial _{\eta
}f(u)\partial _{y}^{3}g(u)+cycle
\end{eqnarray*}

\subsubsection{$J_{2,3}$ }
\begin{eqnarray*}
J_{2,3} &=&(1-N_{\xi })h(z) \int du\theta (x-y)[\partial _{\eta }g(u)\partial
_{y}^{3}f(u)-\partial _{\eta }f(u)\partial _{y}^{3}g(u)]+
\\
&&  +(1-N_{\xi })h(z) \left[ x\{\partial _{\xi
}\partial _{x}^{2}f(z)\}\partial _{\xi }\partial _{x}g(z)-x\{\partial _{\xi
}\partial _{x}f(z)\}\partial _{\xi }\partial _{x}^{2}g(z)
\right]
+cycle
\end{eqnarray*}

\subsubsection{$J_{1,4}$ }
\begin{eqnarray*}
J_{14}=\Delta h(z)\;\int du\eta \partial
_{\eta }g(u)\partial _{y}^{3}\partial _{\eta }f(u)+cycle
\end{eqnarray*}




\subsubsection{$J_{2,4}$ }
\begin{eqnarray*}
J_{24}=2\Delta h(z)\;\int du\eta \theta(x-y)\partial
_{\eta }g(u)\partial _{y}^{3}\partial _{\eta }f(u)+loc+cycle
\end{eqnarray*}

Here we denote as $loc$ some local 2-form. We don't need its specific form
for deriving the deformation.


\subsubsection{$J_{2,6}$ }
$J_{26} =J(m_{0},m_{9})$ where
\begin{equation}\label{m9}
m_{9}(f,g) =\partial _{x}\partial _{\xi }\left[ f,g\right]
\end{equation}

\subsubsection{$J_{2,7}$ }
$J_{27} =J(m_{0},m_{10})$ where
\begin{equation}\label{m10}
m_{10} =\int dz\xi \mu (x)\left( {m}_{2|2}{\ (z|f,g)}-\widetilde{\theta }%
(x)\int du{[\partial }_{\eta }{\ g(u)\partial }_{y}^{3}{\ f(u)-\partial }%
_{\eta }{\ f(u)\partial }_{y}^{3}{\ g(u)]}\right)
\end{equation}

\subsubsection{$J_{3,4}$ }
\begin{eqnarray*}
J_{34}&=&(1-\frac{1}{2}x\partial
_{x})[(1-N_{\xi })f(z)(1-N_{\xi })g(z)]\Delta h(z)+\\
&&(1-N_{\xi
})h\{-\Delta f(z)(1-N_{\xi })\mathcal{E}_{z}g-\Delta g(z)(1-N_{\xi })%
\mathcal{E}_{z}f(z)\}+cycle
\end{eqnarray*}

\subsubsection{$J_{3,5}$ }
\begin{eqnarray*}
J_{35}&=&(1-N_{\xi
})h(z)\int du\partial _{y}f(u)\partial _{y}g(u)+\\
&&\int du\eta
(\{(1-N_{\xi })f\}(1-N_{\xi })g)\partial _{\eta }\partial _{y}^{2}h+cycle
\end{eqnarray*}

\subsubsection{$J_{3,6}$ }
\begin{eqnarray*}
J_{36}&=&\{(1-N_{\xi })h\}\int du\theta
(x-y)\partial _{y}f(u)\partial _{y}g(u)-\\
&&-\int du\theta (x-y)\eta \partial
_{y}(\{(1-N_{\xi })f\}(1-N_{\xi })g)\partial _{y}\partial _{\eta }h+cycle
\end{eqnarray*}

\subsubsection{$J_{3,7}$ }
\begin{eqnarray*}
J_{37}
&=&-\int du\eta \mu (y)\partial _{y}(f\cdot g)\partial _{\eta }h+\\
&&(1-N_{\xi })h(z)\int du\eta
\mu (y)(\partial _{\eta }f\partial _{y}g+\partial _{\eta }g\partial
_{y}f)+cycles
\end{eqnarray*}

\subsection{Relations between constants}

From these values of Jacobiators the next theorem follows
\begin{theorem}\label{cc}
{\it
Let $V_2=m_{2|0}+\mathop{\sum}\limits_{i=1}^6 c_i m_{2|i}+m_{2|7}(M)$.
Then the equation for
2-form $X_2$
\be
J(m_{2|0},X_2)=J(V_2, V_2)
\ee
has the solution if and only if
\begin{eqnarray}
&&c_4 c_i=0, \ \ \ i\ne 4,\label{4}\notag\\
&&c_4 M=0,  \notag\\
&&c_3 c_i=0,\notag\\
&&c_3 M=0.   \notag\label{3}
\end{eqnarray}
}
\end{theorem}

In addition, as $c_1$, $c_2$, $c_3$ are odd Grassmannian elements,
we have $c_i c_i=0$ for $i=1,2,3$.

It follows from this theorem that we have no needs to calculate the Jacobiators
that have the products $c_4 c_i$ ($i\ne4$) or $c_3 c_i)$ as a coefficients.
Namely we have not to calculate
$J_{4,8}$, $J_{4,9}$, $J_{4,10}$, $J_{3,8}$, $J_{3,9}$ and $J_{3,10}$.
As well, because $c_2c_2=0$ we have not to calculate
$J_{2,9}$, $J_{2,10}$, $J_{9,9}$, $J_{9,10}$, $J_{10,10}$.

\subsection{Jacobiators $J_{8,i}, \ J_{9,i}, \ J_{10,i}$.}

The Jacobiators of each pair of 2-forms $J_{i,j}=J(m_{2|i}, m_{2|j})$ ($i=8,9,10$, $j=1,...,10$) are described
in the following table.
\vskip 2mm

\begin{tabular}{r cccccccccc}
$i\setminus j$ \vline
           & 1 & 2 & 3 & 4 & 5 & 6 & 7 & 8 & 9                      & 10\\
\hline
8  \vline  & 0 & a & X & X & 0 & 0 & 0 & 0 & $J(m_{2|0}, m_{2|11}(M))$ & 0\\
9  \vline  & 0 & X & X & X & 0 & 0 & -a & $J(m_{2|0}, m_{2|11}(M))$ & X & X\\
10 \vline  & 0 & X & X & X & 0 & 0 & 0 & 0 & X & X
\end{tabular}
\vskip 3mm

The expression for $a$ can be written, but we don't need it here.
An important relation can be easily verified
\begin{equation}
J_{8,2}+J_{9,7}+J_{10,6}=0
\end{equation}

The 2-form $m_{2|11}$ has the form

\begin{eqnarray}
m_{2|11}(z|f,g)=-\int du\eta \widetilde{\mu}(y)[\partial _{\eta
}f(u)\partial _{\eta }g^{\prime \prime \prime }(u)-\partial _{\eta
}f^{\prime \prime \prime }(u)\partial _{\eta }g(u)].
\end{eqnarray}

\subsection{$J_{i,11}$}
At last,
$J_{i,11}$ are presented in the following table
\vskip 2mm

\begin{tabular}{r ccccccccccc}
$i\setminus j$ \vline
           & 1 & 2 & 3 & 4 & 5 & 6 & 7 & 8 & 9                      & 10 & 11\\
\hline
11 \vline  & 0 & X & X & X & 0 & 0 & 0 & 0 & X & X & X
\end{tabular}
\vskip 3mm

\section{Deformations of antibrackets (Results)}\label{def-res}

\begin{theorem}
\label{-th1}
{\it

Let
\begin{equation}\label{calm}
{\cal M}_{c_4}(f,g)=(-1)^{\varepsilon (f)}{\ \{}%
\frac{ c_4}{1+ c_4 N_{z}/2}{\ \Delta f\}}\mathcal{E}_{z}{\ g+\{}%
\mathcal{E}_{z}{\ f\}}\frac{ c_4}{1+ c_4 N_{z}/2}{\ \Delta g}
\end{equation}%
where $c_4$ is a series in $\hbar$ and $\hbar^2 \theta_i\theta_j$ such that $c_4|_{\hbar=0}=0$.

Let $\widetilde{\theta }$ be smooth function such that $\widetilde{\theta }(x)-{\theta }(x)$
has compact support.

Let $M\in
%
\!\!\!
{\raise1pt\hbox{
$D'$
}}
\!\!\!\!
\big/
\!\!\!
{\raise-3pt\hbox{
$E'$
}}
\!\!
[[\hbar, \hbar\theta_k]]$ be even in $\theta$s and $M|_{\hbar=0}=0$
and let $\mu$ be the kernel of some representative of $M$.


Let
\begin{eqnarray}
&&\ m_{2|0}(z|f,g)=[f,g],\\
&&{\ m}_{2|1}{\ (z|f,g)=}\int {\ du\partial }_{\eta }{\ g(u)\partial }%
_{y}^{3}{\ f(u),\;\;
\epsilon }_{m_{2|1}}{\ =1,}
\label{-5.2.7.1} \\
&&{\ m}_{2|2}{\ (z|f,g)=}\int {\ du\theta (x-y)[\partial }_{\eta }{\
g(u)\partial }_{y}^{3}{\ f(u)-\partial }_{\eta }{\ f(u)\partial }_{y}^{3}{\
g(u)]+}  \notag \\
&&{\ +x[\{\partial }_{\xi }{\ \partial }_{x}^{2}{\ f(z)\}\partial }_{\xi }{\
\partial }_{x}{\ g(z)-\{\partial }_{\xi }{\ \partial }_{x}{\ f(z)\}\partial }%
_{\xi }{\ \partial }_{x}^{2}{\ g(z)],\;\;\epsilon }_{m_{2|2}}{\ =1,}
\label{-5.2.8-1}
\\
&&{\ m}_{2|3}{\ (z|f,g)=(-1)}^{\varepsilon (f)}{\ \{(1-N}_{\xi }{\
)f(z)\}(1-N}_{\xi }{\ )g(z),\;\;\epsilon }_{m_{2|3}}{\ =1,}  \label{-6.3a.1}
\\
&&{\ m}_{2|5}{\ (z|f,g)=}\int {\ du(-1)}^{\epsilon (f)}{\ \partial }_{y}{\
f(u)\partial }_{y}{\ g(u),\;\;\epsilon }_{m_{2|5}}{\ =0,}  \label{-5.2.7-5.1}
\\
&&{\ m}_{2|6}{\ (z|f,g)=}\int {\ du\theta (x-y)(-1)}^{\epsilon (f)}
{\partial }_{y}{\ f(u)\partial }_{y}{\ g(u),\;\;\epsilon }_{m_{2|6}}{\ =0,}
\label{-5.2.7-6.1} \\
&&{\ m}_{2|7}{(M|f,g)}{=}\int {\ du}{\, \eta }\, \mu (y)%
\,[ f(u),g(u)],\;\;\epsilon_{m_{2|7}}=0,
\\
&&\ m_{2|8}(M|f,g) =
\int du
\widetilde \mu(y)
(-1)^{\epsilon (f)}\partial _{y}f(u)\partial _{y}g(u)
\\
&&\ m_{2|9}(z|f,g) =\partial _{x}\partial _{\xi }\left[ f,g\right],
\\
&&\ m_{2|10}(M|f,g) =\int dz\xi \mu (x)\left({\phantom {\int}}\!\!\!\!\!\!\!
{m}_{2|2}{\ (z|f,g)}-
\right. \notag\\
&& \ \ \ \ \ \ \ \ \ \ \ \ \ \ \ \ \  \left.-\widetilde{\theta }%
(x)\int du{[\partial }_{\eta }{\ g(u)\partial }_{y}^{3}{\ f(u)-\partial }%
_{\eta }{\ f(u)\partial }_{y}^{3}{\ g(u)]}\right),
\\
&&\ m_{2|11}(M|f,g) =\int du
\widetilde \mu(y)
\left(
\partial_{\eta } \ g(u)\partial_{y}^{3}\,f(u)-
\partial_{\eta }\, f(u)\partial_{y}^{3}\, g(u)
\right).
\end{eqnarray}

Let $c_i\in\K[[\hbar, \hbar\theta_k]]$ ($i=1,...,6$) be formal series on even parameter $\hbar$
and Grassmannian parameters $\theta_k$; $c_1$, $c_2$, $c_3$ are odd and other $c_i$ are even;
$c_i|_{\hbar=0}=0$.
Let $c_i$ satisfy the relations
\begin{eqnarray}
&&c_4 c_i=0, \ \ \ i\ne 4,\label{-4}\\
&&c_4 M=0,\\
&&c_3 c_i=0,\\
&&c_3 M=0. \label{-3}
\end{eqnarray}

Then the bilinear form
\begin{equation}\label{defform}
C=m_{2|0}+{\cal M}_{c_4}+\sum_{i=1,2,3,5,6}c_i m_{2|i} + m_{2|7}(M)
- c_6  m_{2|8}(M)- c_2 c_6 m_{2|9}- c_2  m_{2|10}(M)-  c_2 c_6^2 m_{2|11}(M)
\end{equation}
is a deformation of antibracket $m_{2|0}$ on $\mathbf {DE}$.
}\end{theorem}

To prove this theorem it is sufficient to check Jacoby identity
for $C$ using previous results for Jacobiators.

The solution (\ref{defform}) is complete in the cases listed in the following theorem.

\begin{theorem}\label{-th2}
{\it
Let $C$ be some deformation of antibracket $m_{2|0}$  on $\mathbf {DE}$
depending on even parameter $\hbar$ and on $\hbar \theta_j$ with
Grassmannian generating elements $\theta_j$.
So, $C=\sum_0^\infty \hbar^p C_p(\hbar\theta_k)$ and $C|_{\hbar=0}=m_{2|0}$.
Let $\sum_1^s \hbar^p C_p(\hbar\theta_k)$
depends on generating elements
of $\G^{n_s}$ ($n_1 \le n_2 \le n_3\le ...$).

Then
\begin{enumerate}
\item\label{i1}
If $C|_{\theta_k=0}\ne m_{2|0}$,
then $C$ has the form (\ref{defform}), and $c_i$ and $M$ satisfy (\ref{-4})-(\ref{-3}).
\item\label{i2}
If $C|_{\theta_k=0}= m_{2|0}$ and $n_s\le 3$ for all $s$,
then $C$ has the form (\ref{defform}), and $c_i$ and $M$ satisfy (\ref{-4})-(\ref{-3}).
\end{enumerate}
}\end{theorem}

{\bf Corollary}\label{cor}
{\it
If $n_s\le 3$ for all $s$,
then $C$ has the form (\ref{defform}), and $c_i$ and $M$ satisfy (\ref{-4})-(\ref{-3}).
}

The proof of these theorems can be found in the Appendices.


\subsection{List of possible deformations}
\label{listdef}

The relations (\ref{-4})-(\ref{-3}) can be solved under the conditions of each item
of Theorem \ref{-th2}.

Item (\ref{i1}) of this Theorem
splits onto 2 possibilities
\begin{enumerate}
\item{$c_4|_{\theta=0}\ne 0$.}

In this case
 $c_1=c_2=c_3=0$, $c_5=c_6=0$, $M=0$ and
\begin{equation}\label{defform1}
C=m_{2|0}+{\cal M}_{c_4}
\end{equation}
\item{At least one of $c_5|_{\theta=0}$, $c_6|_{\theta=0}$, $M|_{\theta=0}$ is not zero.}

In this case $c_3=0$, $c_4=0$ and
\begin{equation}\label{defform2}
C=m_{2|0}+\sum_{i=1,2,5,6}c_i m_{2|i} + m_{2|7}(M)
- c_6  m_{2|8}(M)- c_2 c_6 m_{2|9}- c_2  m_{2|10}(M)-  c_2 c_6^2 m_{2|11}(M)
\end{equation}
\end{enumerate}

Now, let us look at the possibilities, that follows from
Item (\ref{i2}) of Theorem \ref{-th2}.

In this case $c_4=\hbar^2 (c_{4,1} \theta_2 \theta_3+ c_{4,2} \theta_3 \theta_1 + c_{4,3}\theta_1 \theta_2)$,
where $c_{4,k} \in \K[[\hbar]]$. Clearly, it is possible to choose a new generators of $\G^3$,
\be\label{newtheta3}
\theta_k^\prime=\sum_{l=1,2,3}a_{k,l}(\hbar)\theta_l
\ee
in such a way that
$c_4=\hbar^2 c_{4,3}^\prime (\hbar)\theta_1^\prime \theta_2^\prime$.

So, it is sufficiently to consider the coefficient $c_4$ in the form
$c_4=\hbar^2 \alpha_4 \theta_1 \theta_2$ where $d_4\in \hbar^2K[[\hbar]]$.

This form of $c_4$ is preserved under the following changes of generators of $\G^3$:
\be\label{newtheta2}
\left\{
\begin{array}{l}
\theta_k^\prime=\sum_{l=1,2}a_{k,l}(\hbar)\theta_l\quad k=1,2\\
\theta_3^\prime=\theta_3.
\end{array}
\right.
\ee

Up to the transformation of the form (\ref{newtheta3}), Item (\ref{i2}) of Theorem \ref{th2}
gives two possibilities:

\begin{enumerate}

\item

{$c_i=\hbar\sum_{j=1,2}\beta_i^j\theta_j+\hbar^3\gamma_i \theta_1\theta_2\theta_3$,
($i=1,2$), $c_3=\hbar^3\gamma_3\theta_1\theta_2\theta_3$,
\\
$c_{4}=\hbar^2 \alpha_{4}\theta_1\theta_2$,
$c_{5,6}=\hbar^2 \sum_{k,l,m=1,2,3}\alpha^k_{5,6}\varepsilon_{klm}\theta_l\theta_m$,
$M=\hbar^2 \sum_{k,l,m=1,2,3}M_k\varepsilon_{klm}\theta_l\theta_m$,}

where $\alpha_4,\ \alpha_k^i,\ \beta_i, \gamma_i\in \K[[\hbar]]$,
$M_k\in
%
\!\!\!
{\raise1pt\hbox{
$D'$
}}
\!\!\!\!
\big/
\!\!\!
{\raise-3pt\hbox{
$E'$
}}
\!\!
[[\hbar]]$

In this case
\begin{equation}\label{defform4}
C=m_{2|0}+\sum_{i=1,2,3,4,5,6}c_i m_{2|i} + m_{2|7}(M)
- c_2 c_6 m_{2|9}- c_2  m_{2|10}(M)
\end{equation}

\item{$c_i=\theta_1 \tilde c_i$, $M=\theta_1 \tilde M$}.

In this case
\begin{equation}\label{defform3}
C=m_{2|0}+\theta_1\left(
\sum_{i=1,2,3,4,5,6}\tilde c_i m_{2|i} + m_{2|7}(\tilde M)
\right).
\end{equation}

\end{enumerate}





\setcounter{MaxMatrixCols}{10}

\def\theequation{A\arabic{appen}.\arabic{equation}}

\renewcommand{\theorem}{\par\refstepcounter{theorem}
{\bf Theorem A\arabic{appen}.\arabic{theorem}. }}
\renewcommand{\lemma}{\par\refstepcounter{lemma}
{\bf Lemma A\arabic{appen}.\arabic{lemma}. }}
\renewcommand{\proposition}{\par\refstepcounter{proposition}
{\bf Proposition A\arabic{appen}.\arabic{proposition}. }}
\makeatletter \@addtoreset{theorem}{appen}
\makeatletter \@addtoreset{lemma}{appen}
\makeatletter \@addtoreset{proposition}{appen}
\makeatletter \@addtoreset{equation}{appen}
\renewcommand\thetheorem{A\theappen.\arabic{theorem}}
\renewcommand\thelemma{A\theappen.\arabic{lemma}}
\renewcommand\theproposition{A\theappen.\arabic{proposition}}

\appen{Jacobiator of cocycle and coboard}\label{app1}


Here we formulate and prove Lemma, which helps to
calculate Jacobiators $J_{i,7}$ in Section \ref{jac-res}.

\begin{lemma}
{\it
Let $\epsilon $ be some parity, preserving by $m_{2|0}=[f,g],$ i.e. $%
\epsilon ([f,g])=\epsilon (f)+\epsilon (g)$

Let 2-form $m_{2}(f,g)$ be cocycle, i.e.

\begin{eqnarray*}
&&\,{(-1)}^{\epsilon (f)\epsilon (h)}m_{ 2  }([f,g],h)+\mathrm{cycle}{%
(f,g,h)}= \\
&&\,=-{(-1)}^{\epsilon (f)\epsilon (h)}[m_{ 2  }(f,g),h]+\mathrm{cycle}{%
(f,g,h)},
\end{eqnarray*}

and let $n{     }(f,g)=d_{   1  }M_{1}(f,g)$, $\epsilon _{n{     }}=\epsilon _{M_{1}}$.
Then%
\[
J(m_{ 2  },d_{  1   }M_{1})=J(m_{2|0},U_{2})
\]

where%
\begin{eqnarray*}
&&U_{2}={(-1)}^{\epsilon _{m_{ 2  }}\epsilon
_{M_{1}}}M_{1}(m_{ 2  }(f,g))-m_{ 2  }(M_{1}(f),g)+ \\
&&\,+{(-1)}^{\epsilon (f)\epsilon (g)}m_{ 2  }(M_{1}(g),f),\;\epsilon
_{U_{2}}=\epsilon _{m_{ 2  }}+\epsilon _{M_{1}}.
\end{eqnarray*}
}
\end{lemma}

\begin{proof}
Let, for shortness, $R=J(m_{2|0},U_{2})$,
$J=J(m_{ 2  },d_{  1   }M_{1})$,
$\sigma (f,h)={(-1)}%
^{\epsilon (f)\epsilon (h)}$.

Evaluate the difference $J-R$:
\begin{eqnarray*}
{\ } &&J-R{=}\sigma (f,h)\{m_{ 2  }([M_{1}(f),g],h)-\sigma
(g,f)m_{ 2  }([M_{1}(g),f],h)- \\
&&-{       }{m}_{2}{       }{(M}_{1}{       }{([f,g])\},h)+(-1)}^{\epsilon
_{m_{ 2  }}\epsilon _{M_{1}}}{       }{[M}_{1}{       }{(m}_{2}{       }{(f,g)),h]}+ \\
&&+(-1)^{\epsilon _{M_{1}}[\epsilon (f)+\epsilon (g)]}[m_{ 2  }(f,g),M_{1}(h)]-{%
(-1)}^{\epsilon _{m_{ 2  }}\epsilon _{M_{1}}}M_{1}([m_{ 2  }(f,g),h])- \\
&&-{       }{(-1)}^{\epsilon _{m_{ 2  }}\epsilon _{M_{1}}}{       }{[M}_{1}{       }{(m}_{2}%
{       }{(f,g)),h]+[}m_{ 2  }((M_{1}(f),g),h]- \\
&&-\sigma (g,f){       }{[}m_{ 2  }((M_{1}(g),f),h]- \\
&&-(-1)^{\epsilon _{m_{ 2  }}\epsilon _{M_{1}}}M_{1}(m_{ 2  }([f,g],h))+{       }{m}%
_{2}{       }{(M}_{1}{       }{([f,g]),h)-} \\
&&\,-\sigma (f,h)\sigma (g,h)m_{ 2  }(M_{1}(h),[f,g])\}+\mathrm{cycle}%
\;\Longrightarrow
\end{eqnarray*}%

so
\begin{eqnarray*}
{\ } &&J-R{=}\sigma (f,h)\{m_{ 2  }([M_{1}(f),g],h)-\sigma
(g,f)m_{ 2  }([M_{1}(g),f],h)+ \\
&&+(-1)^{\epsilon _{M_{1}}[\epsilon (f)+\epsilon (g)]}[m_{ 2  }(f,g),M_{1}(h)]-%
{       }{(-1)}^{\epsilon _{m_{ 2  }}\epsilon _{M_{1}}}{       }{M}_{1}{       }{([m}_{2}%
{       }{(f,g),h])}+ \\
&&+[m_{ 2  }((M_{1}(f),g),h]-\sigma (g,f)[m_{ 2  }((M_{1}(g),f),h]- \\
&&{       }{-(-1)}^{\epsilon _{m_{ 2  }}\epsilon _{M_{1}}}{       }{M}_{1}{       }{(m}_{2}%
{       }{([f,g],h))-}\sigma (f,h)\sigma (g,h)m_{ 2  }(M_{1}(h),[f,g])\}+ \\
&&\,+\mathrm{cycle}
\end{eqnarray*}%
and
\begin{eqnarray*}
{\ } &&J-R{=}\sigma (f,h)\{m_{ 2  }([M_{1}(f),g],h)-\sigma
(g,f)m_{ 2  }([M_{1}(g),f],h)+ \\
&&+(-1)^{\epsilon _{M_{1}}[\epsilon (f)+\epsilon
(g)]}[m_{ 2  }(f,g),M_{1}(h)]+[m_{ 2  }((M_{1}(f),g),h]- \\
&&-\sigma (g,f)[m_{ 2  }((M_{1}(g),f),h]-\sigma (f,h)\sigma
(g,h)m_{ 2  }(M_{1}(h),[f,g])\} \\
&&+\mathrm{cycle},
\end{eqnarray*}%
where we use the relation%
\begin{equation}
{\sigma (f,h)\{}m_{ 2  }([f,g],h)+[m_{ 2  }(f,g),h]\}+\mathrm{cycle}=0.  \label{1}
\end{equation}

After proper cycling we obtain
\begin{eqnarray*}
&&J-R{=}\sigma (g,h)m_{ 2  }([M_{1}(h),f]),g)-\sigma (f,g)\sigma
(g,h)m_{ 2  }([M_{1}(h),g],f)+ \\
&&+\sigma (f,h){(-1)}^{\epsilon _{M_{1}}[\epsilon (f)+\epsilon
(g)]}[m_{ 2  }(f,g),M_{1}(h)]+\sigma (g,h){       }{[}m_{ 2  }(M_{1}(h),f),g]- \\
&&-\sigma (f,g)\sigma (g,h){       }{[}m_{ 2  }(M_{1}(h),g),f]-\sigma
(g,h)m_{ 2  }(M_{1}(h),[f,g])+\mathrm{cycle}= \\
&&\,=(-1)^{\epsilon _{M_{1}}\epsilon (g)}\{\sigma
(g,M_{ 1 }(h))m_{ 2  }([M_{ 1 }(h),f]),g)+\sigma (f,g)m_{ 2  }([M_{ 1 }(h),g],f)+ \\
&&+\sigma (f,M_{ 1 }(h))[m_{ 2  }(f,g),M_{ 1 }(h)]+\sigma (g,M_{ 1 }(h)){       }{[}%
m_{ 2  }(M_{ 1 }(h),f),g]+ \\
&&+\sigma (f,g){       }{[}m_{ 2  }(g,M_{ 1 }(h)),f]+\sigma
(f,M_{ 1 }(h))m_{ 2  }([f,g],M_{ 1 }(h))\}+ \\
&&\,+\mathrm{cycle},
\end{eqnarray*}%
Because $m$ is a cocycle, the last expession is zero before cycling:
\bee
&&(-1)^{\epsilon _{M_{1}}\epsilon (g)}\{\sigma
(g,M_{ 1 }(h))m_{ 2  }([M_{ 1 }(h),f]),g)+\sigma (f,g)m_{ 2  }([M_{ 1 }(h),g],f)+ \notag\\
&&+\sigma (f,M_{ 1 }(h))[m_{ 2  }(f,g),M_{ 1 }(h)]+\sigma (g,M_{ 1 }(h)){       }{[}%
m_{ 2  }(M_{ 1 }(h),f),g]+ \notag\\
&&+\sigma (f,g){       }{[}m_{ 2  }(g,M_{ 1 }(h)),f]+\sigma
(f,M_{ 1 }(h))m_{ 2  }([f,g],M_{ 1 }(h))\} =0 \notag
\eee
So, $J-R=0$.
\end{proof}


\appen{Jacoby Identity}\label{-app1}


We consider the deformation $C$ as a formal series on deformation parameter $\hbar$,
each term of this series depends in its turn on generators of some Grassmann algebra $\G^{s_k}$.

\qquad
\begin{eqnarray*}
&&C(z|f,g)=C^{(0)}(z|f,g)+\hbar C^{(1)}(z|f,g)+\hbar
^{2}C^{(2)}(z|f,g)+\hbar^{3}C^{(3)}(z|f,g)+..., \\
&&C^{(k)}(z|f,g)=\theta ^{k}\widetilde C^{(k)}(z|f,g),
\end{eqnarray*}%
The notation $\theta^k$ means $\theta_{i_1}\theta_{i_2}...\theta_{i_k}$,
The forms $\widetilde C^{(k)}$ depend on $\hbar$ and don't depend on $\theta$
\begin{eqnarray*}
&&\varepsilon _{C}=1,\;\varepsilon _{C^{(0)}}=\varepsilon
_{C^{(p)}}=1,\;\varepsilon _{A_{\ast }^{(p)}}=1+p, \\
&&C(z|g,f)=(-1)^{\varepsilon (f)\varepsilon (g)+\varepsilon (f)+\varepsilon
(g)}C(z|f,g),\; \\
&&C^{(0)}(z|g,f)=(-1)^{\varepsilon (f)\varepsilon (g)+\varepsilon
(f)+\varepsilon (g)}C^{(0)}(z|f,g), \\
&&C^{(p)}(z|g,f)=(-1)^{\varepsilon (f)\varepsilon (g)+\varepsilon
(f)+\varepsilon (g)}C^{(p)}(z|f,g),
\end{eqnarray*}

Further,
$$C^{(k)}(z|f,g)=\sum_l \hbar^l C^{(k)}_l(z|f,g).$$

Jacoby identity
\begin{equation}
(-1)^{\varepsilon (f)\varepsilon (g)+\varepsilon (f)+\varepsilon
(g)}C(z|C(|f,g),h)+\mathrm{cycle}(f,g,h)=0  \label{2.1}
\end{equation}%
leads to
the equations:%
%
\begin{eqnarray}
&&\!\!\!\!\!(-1)^{\varepsilon (f)\varepsilon (h)+\varepsilon (f)+\varepsilon
(h)}C^{(0)}(z|C^{(0)}(|f,g),h)+\mathrm{cycle}(f,g,h)=0,  \label{1.1a} \\
&&\!\!\!\!\!(-1)^{\varepsilon (f)\varepsilon (h)+\varepsilon (f)+\varepsilon (h)}
\left[ C^{(0)}(z|C^{(1)}(|f,g),h)+C^{(1)}(z|C^{(0)}(|f,g),h)\right] +\mathrm{%
cycle}=0,  \label{1.1b}
\end{eqnarray}%
and so on.

\appen{Solution of eq. (\protect\ref{1.1a})}

In what follows we use the notation $\omega$ both for the linear functional $M\in D'/E'$ defining
cocycle $mn_{2|7}$ in Theorem \ref{m7} and for the kernel of this functional.

\subappen{First order}

Represent $C^{(0)}(z|f,g)$ in the form%
\begin{equation*}
C^{(0)}(z|f,g)=m_{2|0}(z|f,g)+\hbar C_{1}^{(0)}(z|f,g)+O(\hbar
^{2}),\;\varepsilon _{C_{1}^{(0)}}=1.
\end{equation*}%
It follows from eq. (\ref{1.1a}) that%
\begin{equation*}
J(m_{2|0},C_{1}^{(0)};\,z|f,g,h)=0.
\end{equation*}%
This equation was solved \ in Section \ref{cohom-res} and its general solution in the
sector $\varepsilon _{C_{1}^{(0)}}=1$ is (up to exact forms)%
\begin{eqnarray*}
&&C_{1}^{(0)}(\,z|f,g)=c_{a|1}m_{2|a}(\,z|f,g)+m_{2|7|1}(\,z|f,g),\;a=4,5,6,
\\
&&m_{2|7|k}(\,z|f,g)=\left. m_{2|7}(\,z|f,g)\right| _{\omega \rightarrow
\omega _{k}}.
\end{eqnarray*}

The function $C^{(0)}(z|f,g)$ can be represented in the form%
\begin{eqnarray*}
&&C^{(0)}(z|f,g)=\mathcal{N}_{c_{4|1}}(z|f,g)+\hbar
c_{b|1}m_{2|b}(z|f,g)+\hbar m_{2|7|1}(\,z|f,g)+ \\
&&\,+\hbar ^{2}C_{2}^{(0)}(z|f,g)+O(\hbar ^{3}),\;b=5,6,\;\varepsilon
_{C_{2}^{(0)}}=1,
\end{eqnarray*}
where
\be
\mathcal{N}_c= m_{2|0}+{\cal M}_c,
\ee
and ${\cal M}_c$ is defined by (\ref{calm}).

\subappen{ Second order}

It follows from eq. (\ref{1.1a}) that%
\begin{eqnarray}
&&J(m_{2|0},\tilde{C}%
_{2}^{(0)})+c_{4|1}c_{b|1}J(m_{2|4},m_{2|b})+c_{4|1}J_{1}(m_{2|4},m_{2|7|1})=0,
\label{3.2.0} \\
&&b=5,6,\;\tilde{C}_{2}^{(0)}=C_{2}^{(0)}+c_{6|1}C_{67|1}+c_{4|1}C_{47|1},
\notag \\
&&C_{67|k}=\left. C_{67}\right| _{\omega \rightarrow \omega
_{k}},\;C_{47|k}=\left. C_{47}\right| _{\omega \rightarrow \omega _{k}}.
\notag
\end{eqnarray}%
Let us remind that {\large
\begin{eqnarray*}
&&J(m_{2|4},m_{2|a};\,z|\hat{f}_{0},\hat{g}_{0},\hat{h}%
_{0})=J(m_{2|4},m_{2|a};\,z|\hat{f}_{0},\hat{g}_{0},\hat{h}_{1})= \\
&&\,=J(m_{2|4},m_{2|a};\,z|\hat{f}_{1},\hat{g}_{1},\hat{h}_{1})\equiv
0,\;a=5,6,7, \\
&&C_{67}(z|\hat{f}_{0},\hat{g}_{0})=C_{67}(z|\hat{f}_{1},\hat{g}%
_{1})=C_{47|2}(\hat{f}_{0},\hat{g}_{0})=C_{47|2}(\hat{f}_{1},\hat{g}%
_{1})\equiv 0.
\end{eqnarray*}%
}

Represent the form $M(z|f)$, $\varepsilon _{M}=0$, in the form%
\begin{eqnarray*}
M(z|f) &=&\int du[\eta M_{(0,1)}(x|y)-\xi M_{(1,0)}(x|y)]f(u), \\
M(z|\hat{f}_{0}) &=&\int dyM_{(0,1)}(x|y)f_{0}(y)\equiv M_{(0,1)}(x|f_{0}),\;
\\
M(z|\hat{f}_{1}) &=&\xi \int dyM_{(1,0)}(x|y)f_{1}(y)\equiv \xi
M_{(1,0)}(x|f_{1}).
\end{eqnarray*}

Represent the form $C(z|f,g)=-(-1)^{(\varepsilon (f)+1)(\varepsilon
(g)+1)}C(z|g,f)$, $\varepsilon _{C}=1$, $)$ in the form%
\begin{eqnarray*}
&&C(z|f,g)=\int dvdu[-\xi C_{(1)}(x|y_{u},y_{v})+\eta
_{u}C_{(2)}(x|y_{u},y_{v})-\eta _{v}C_{(2)}(x|y_{v},y_{u})+ \\
&&\,+\xi \eta _{u}\eta _{v}C_{(3)}(x|y_{u},y_{v})](-1)^{\varepsilon
(f)}f(u)g(v),\;C_{(3)}(x|y_{v},y_{u})=C_{(3)}(x|y_{u},y_{v}), \\
&&C_{(1)}(x|y_{v},y_{u})=-C_{(1)}(x|y_{u},y_{v}),
\end{eqnarray*}%
\begin{eqnarray*}
&&C(z|\hat{f}_{0},\hat{g}_{0})=\xi C_{(3)}(x|f_{0},g_{0}),\;C(z|\hat{f}_{0},%
\hat{g}_{1})=C_{(2)}(x|f_{0},g_{1}) \\
&&C(z|\hat{f}_{1},\hat{g}_{0})=-C_{(2)}(x|g_{0},f_{1}),\;C(z|\hat{f}_{1},%
\hat{g}_{1})=\xi C_{(1)}(x|f_{1},g_{1}),
\end{eqnarray*}%
\begin{eqnarray*}
&&C_{67}(z|\hat{f}_{0},\hat{g}_{0})=0\;\Longrightarrow
\;C_{67(3)}(x|f_{0},g_{0})=0, \\
&&C_{67}(z|\hat{f}_{1},\hat{g}_{1})=0\;\Longrightarrow
\;C_{67(1)}(x|f_{1},g_{1})=0, \\
&&C_{67}(z|\hat{f}_{0},\hat{g}_{1})=C_{67(2)}(x|f_{0},g_{1})=\int dy\tilde{%
\omega}(y)f_{0}^{\prime }(y)g_{1}^{\prime }(y), \\
&&C_{67}(z|\hat{f}_{1},\hat{g}_{0})=-C_{67(2)}(x|g_{0},f_{1})..
\end{eqnarray*}

We have for $M_{d_{1}}(z|f,g)=d_{1}^{\mathrm{ad}}M(f,g)$:%
\begin{eqnarray}
&&M_{d_{1}(3)}(z|f_{0},g_{0})=0,  \label{5.2.1a} \\
&&M_{d_{1}(2)}(x|f_{0},g_{1})=\{\partial
_{x}M_{(0,1)}(x|f_{0})\}g_{1}(x)+M_{(1,0)}(x|g_{1})f_{0}^{\prime
}(x)-M_{(0,1)}(x|f_{0}^{\prime }g_{1}),  \label{5.2.1b} \\
&&M_{d_{1}(1)}(x|f_{1},g_{1})=\{\partial
_{x}M_{(1,0)}(x|f_{1})\}g_{1}(x)-\{\partial
_{x}M_{(1,0)}(x|g_{1})\}f_{1}(x)+M_{(1,0)}(x|g_{1})f_{1}^{\prime }(x)-
\notag \\
&&-M_{(1,0)}(x|f_{1})\}g_{1}^{\prime }(x)+M_{(1,0)}(x|f_{1}g_{1}^{\prime
}-f_{1}^{\prime }g_{1}),  \notag
\end{eqnarray}

It follows from $J(m_{2|0},\tilde{C}_{2}^{(0)};z|\hat{f}_{0},\hat{g}_{0},%
\hat{h}_{0})=0$
\begin{equation*}
C_{2(3)}^{(0)}(x|f_{0},g_{0})h_{0}^{\prime }(x)+\mathrm{cycle}%
(f_{0},g_{0},h_{0})=0,
\end{equation*}%
This equation was solved in \cite{JMP} and \cite{apoi}  and we have%
\begin{equation*}
C_{2(3)}^{(0)}(x|f_{0},g_{0})=0.
\end{equation*}

After this, we have $J(m_{2|0},\tilde{C}_{2}^{(0)};z|\hat{f}_{0},\hat{g}_{0},%
\hat{h}_{1})\equiv 0$.

It follows from $J(m_{2|0},\tilde{C}_{2}^{(0)};z|\hat{f}_{1},\hat{g}_{1},%
\hat{h}_{1})=0$%
\begin{eqnarray}
&&C_{2(1)}^{(0)}(x|f_{1},g_{1})h_{1}^{\prime }(x)-\{\partial
_{x}C_{2(1)}^{(0)}(x|f_{1},g_{1})\}h_{1}(x)+  \notag \\
&&\,+C_{2(1)}^{(0)}(x|f_{1}g_{1}^{\prime }-f_{1}g_{1}^{\prime },h_{1})+%
\mathrm{cycle}(f_{1},g_{1},h_{1})=0,  \label{3.2.1c}
\end{eqnarray}%
Here we use the property $C_{67(3)}=C_{67(1)}=M_{47|2(3)}=M_{47|2(1)}=0$,
such that we have $\tilde{C}_{2(3)}^{(0)}=C_{2(3)}^{(0)}$, $\tilde{C}%
_{2(1)}^{(0)}=C_{2(1)}^{(0)}$. Eq. (\ref{3.2.1c}). was also solved in \cite{JMP} and \cite{apoi}
and we have%
\begin{equation*}
C_{2(1)}^{(0)}(x|f_{1},g_{1})=c_{4|2}m_{2|4(1)}(x|f_{1},g_{1})+M_{d_{1}(1)}(x|f_{1},g_{1}),
\end{equation*}%
where $M(z|f)$ is an one-form.

It follows from eq. (\ref{3.2.0}) for $f=\hat{f}_{0}$, $g=\hat{g}_{1}$, $h=%
\hat{h}_{1}$%
\begin{eqnarray*}
&&\{\partial _{x}\tilde{C}_{2(2)}(x|f_{0},g_{1})\}h_{1}(x)+\tilde{C}%
_{2(2)}(x|f_{0}^{\prime }g_{1},h_{1})+\tilde{C}_{2(2)}(x|f_{0},g_{1}h_{1}^{%
\prime })- \\
&&\;-(g_{1}\leftrightarrow h_{1})-C_{2(1)}(x|g_{1},h_{1})f_{0}^{\prime
}(x)=c_{4|1}c_{b|1}J(m_{2|4},m_{2|b};z|\hat{f}_{0},\hat{g}_{1},\hat{h}_{1})+
\\
&&\,+c_{4|1}J_{1}(m_{2|4},m_{2|7};z|\hat{f}_{0},\hat{g}_{1},\hat{h}%
_{1}),\;b=5,6,
\end{eqnarray*}

or%
\begin{eqnarray*}
&&\{\partial _{x}D_{2(2)}(x|f_{0},g_{1})\}h_{1}(x)+D_{2(2)}(x|f_{0}^{\prime
}g_{1},h_{1})+D_{2(2)}(x|f_{0},g_{1}h_{1}^{\prime })- \\
&&\;-(g_{1}\leftrightarrow
h_{1})=c_{4|2}m_{2|4(1)}(x|g_{1},h_{1})f_{0}^{\prime
}(x)+c_{4|1}c_{b|1}J(m_{2|4},m_{2|b};z|\hat{f}_{0},\hat{g}_{1},\hat{h}_{1})+
\\
&&+c_{4|1}J_{1}(m_{2|4},m_{2|7};z|\hat{f}_{0},\hat{g}_{1},\hat{h}_{1}),
\end{eqnarray*}%
where $D_{2(2)}(x|f_{0},g_{1})=\tilde{C}%
_{2(2)}(x|f_{0},g_{1})-M_{(1,0)}(x|g_{1})f_{0}^{\prime }(x)$

Consider the domain $[x\cup \mathrm{supp}f_{0}]\cap \lbrack \mathrm{supp}%
g_{1}\cup \mathrm{supp}h_{1}]=\varnothing $. In this domain we have $%
J(m_{2|4},m_{2|b})=J_{1}(m_{2|4},m_{2|7})=0$ and we find%
\begin{equation*}
\hat{D}_{2(2)}(x|f_{0},g_{1}h_{1}^{\prime }-g_{1}^{\prime }h_{1})=0.
\end{equation*}

Let $f_{0}\in D$. Then we have
\begin{equation*}
D_{2(2)}(x|f_{0},g_{1})=\sum_{q=0}^{Q}T_{1}^{q}(x|f_{0})\partial
_{x}^{q}g_{1}(x)+\sum_{q=0}^{Q}T_{2}^{q}(x|g_{1}\partial
^{q}f_{0}),\;f_{0}\in D.
\end{equation*}

Consider the domain $\mathrm{supp}f_{0}\cap \lbrack x\cup \mathrm{supp}%
g_{1}\cup \mathrm{supp}h_{1}]=\varnothing $. In this domain we have $%
J(m_{2|4},m_{2|b}^{(0)})=J_{1}(m_{2|4},m_{2|7})=0$ and we find%
\begin{equation*}
\{\partial
_{x}D_{2(2)}(x|f_{0},g_{1})\}h_{1}(x)+D_{2(2)}(x|f_{0},g_{1}h_{1}^{\prime
})-(g_{1}\leftrightarrow h_{1})=0.
\end{equation*}

Let $f_{0}\in D$. Then we have
\begin{eqnarray}
&&\sum_{q=0}^{Q}\{\partial _{x}\hat{T}_{1}^{q}(x|f_{0})\}[\partial
_{x}^{q}g_{1}(x)h_{1}(x)-g_{1}(x)\partial _{x}^{q}h_{1}(x)]+  \notag \\
&&+\sum_{q=0}^{Q}\hat{T}_{1}^{q}(x|f_{0})\}[\partial
_{x}^{q+1}g_{1}(x)h_{1}(x)-g_{1}(x)\partial _{x}^{q+1}h_{1}(x)]+  \notag \\
&&\,+\sum_{q=0}^{Q}\hat{T}_{1}^{q}(x|f_{0})\partial
_{x}^{q}[g_{1}(x)h_{1}^{\prime }(x)-g_{1}^{\prime }(x)h(x)]=0.  \label{3.2.2}
\end{eqnarray}

Let $g_{1}(x)=e^{px}$, $h_{1}(x)=e^{kx}$. Then we have
\begin{equation}
\sum_{q=0}^{Q}\{\partial _{x}(p^{q}-k^{q})+\hat{T}%
_{1}^{q}(x|f_{0})[p^{q+1}-k^{q+1}-(p+k)^{q}(p-k)]\}=0  \label{3.2.2a}
\end{equation}

Considering in eq. (\ref{3.2.2a}) the terms of higher ($Q+1$) order in $p$, $%
k$, we obtain
\begin{equation*}
\hat{T}_{1}^{q}=0,\;T_{1}^{q}=\mathrm{loc},\;q\geq 2,\;f_{0}\in D,
\end{equation*}%
and eq. (\ref{3.2.2}) reduces to the form%
\begin{equation*}
\partial _{x}\hat{T}_{1}^{q}(x|f_{0})=0,\;\partial _{x}T_{1}^{q}(x|f_{0})=%
\mathrm{loc},\;f_{0}\in D\;\Longrightarrow
\end{equation*}%
\begin{eqnarray*}
&&D_{2(2)}(x|f_{0},g_{1})=T_{1}^{0}(x|f_{0})g_{1}(x)+[t_{1|1}^{1}(f_{0})+t_{1|2}^{1}(x|f_{0})]g_{1}^{\prime }(x)+
\\
&&\,+\sum_{q=0}^{Q}T_{2}^{q}(x|g_{1}\partial ^{q}f_{0})+\mathrm{loc}, \\
&&t_{1|1}^{1}(f_{0})=\int
dyt_{1|1}^{1}(y)f_{0}(y),\;t_{1|2}^{1}(x|f_{0})=\int dy\theta
(x-y)t_{1|2}^{1}(y)f_{0}(y),\mathrm{\;}f_{0}\in D.
\end{eqnarray*}

Consider the domain $[x\cup \mathrm{supp}h_{1}]\cap \lbrack \mathrm{supp}%
f_{0}\cup \mathrm{supp}g_{1}]=\varnothing $. In this domain we find%
\begin{eqnarray*}
&&\{\partial _{x}\hat{D}_{2(2)}(x|f_{0},g_{1})\}h_{1}(x)+\hat{D}%
_{2(2)}(x|f_{0}^{\prime }g_{1},h_{1})= \\
&&\,=[c_{4|1}c_{5|1}\overline{f_{0}g_{1}^{\prime \prime }}+c_{4|1}c_{6|1}%
\overline{\theta _{x}f_{0}g_{1}^{\prime \prime }}-c_{4|1}\omega
_{1}(f_{0}^{\prime }g_{1})]h_{1}^{\prime }(x).
\end{eqnarray*}

Let $f_{0}\in D$. Then we find%
\begin{eqnarray}
&&\sum_{q=0}^{Q}\{\partial _{x}\hat{T}_{2}^{q}(x|g_{1}\partial
^{q}f_{0})\}h_{1}(x)+\hat{T}_{1}^{0}(x|f_{0}^{\prime
}g_{1})h_{1}(x)+[t_{1|1}^{1}(f_{0}^{\prime }g_{1})+  \notag \\
&&\,+\hat{t}_{1|2}^{1}(x|f_{0}^{\prime }g_{1})]h_{1}^{\prime
}(x)=[c_{4|1}c_{5|1}\overline{f_{0}g_{1}^{\prime \prime }}+c_{4|1}c_{6|1}%
\overline{\theta _{x}f_{0}g_{1}^{\prime \prime }}-  \notag \\
&&\,-c_{4|1}\omega _{1}(f_{0}^{\prime }g_{1})]h_{1}^{\prime
}(x)\;\Longrightarrow  \label{3.2.3}
\end{eqnarray}

It follows from eq. (\ref{3.2.3})%
\begin{eqnarray}
&&\,\partial _{x}\hat{T}_{2}^{q}(x|g_{1})=0,\;\forall q,\;q\neq 1,  \notag \\
&&\hat{T}_{1}^{0}(x|g_{1})=-\partial _{x}\hat{T}_{2}^{1}(x|g_{1})\;%
\Longrightarrow \;T_{1}^{0}(x|g_{1})=-\partial _{x}T_{2}^{1}(x|g_{1})+%
\mathrm{loc},  \notag \\
&&t_{1|1}^{1}(f_{0}^{\prime }g_{1})+\hat{t}_{1|2}^{1}(x|f_{0}^{\prime
}g_{1})=c_{4|1}c_{5|1}\overline{f_{0}g_{1}^{\prime \prime }}+c_{4|1}c_{6|1}%
\overline{\theta _{x}f_{0}g_{1}^{\prime \prime }}-c_{4|1}\omega
_{1}(f_{0}^{\prime }g_{1}).  \label{3.2.4}
\end{eqnarray}

i) Using $\mathrm{supp}g_{1}>x$, we obtain
\begin{equation*}
t_{1|1}^{1}(f_{0}^{\prime }g_{1})=c_{4|1}c_{5|1}\overline{f_{0}g_{1}^{\prime
\prime }}-c_{4|1}\omega _{1}(f_{0}^{\prime }g_{1})=0
\end{equation*}%
and then%
\begin{equation*}
\hat{t}_{1|2}^{1}(x|f_{0}^{\prime }g_{1})=c_{4|1}c_{6|1}\overline{\theta
_{x}f_{0}g_{1}^{\prime \prime }}.
\end{equation*}

ii) Using $f_{0}(y)=y$ for $y\in \mathrm{supp}g_{1}$, we find
\begin{eqnarray*}
&&\hat{t}_{1|2}^{1}(x|g_{1})=0\;\Longrightarrow \;t_{1|2}^{1}(x|g_{1})=%
\mathrm{loc}\;\Longrightarrow \;t_{1|2}^{1}(x|g_{1})=0, \\
&&c_{4|1}c_{5|1}=c_{4|1}c_{6|1}=0,\;t_{1|1}^{1}(g_{1})=-c_{4|1}\omega
_{1}(g_{1}).
\end{eqnarray*}

Note that%
\begin{eqnarray*}
&&T_{1}^{0}(x|f_{0})g_{1}(x)+T_{2}^{1}(x|g_{1}f_{0}^{\prime })=-\partial
_{x}T_{2}^{1}(x|f_{0})g_{1}(x)+T_{2}^{1}(x|g_{1}f_{0}^{\prime })+ \\
&&\,+\mathrm{loc}=\tilde{M}_{d_{1}(2)}(x|f_{0},g_{1})+\mathrm{loc},\;\tilde{M%
}(x|f_{0})=-\int du\eta T_{2}^{1}(x|y)f(u), \\
&&\tilde{M}_{d_{1}}(x|\hat{f}_{0},\hat{g}_{0})=\tilde{M}_{d_{1}}(x|\hat{f}%
_{1},\hat{g}_{1})=0.
\end{eqnarray*}%
Rewrite the equation for $C_{2}^{(0)}$ just obtained%
\begin{equation}
J(m_{2|0},\tilde{C}_{2}^{(0)})+c_{4|1}J_{1}(m_{2|4},m_{2|7})=0.
\label{3.2.5}
\end{equation}

Consider eq. (\ref{3.2.5}) in more details. Let the functions $f_{0}$, $%
g_{0} $, $h_{0}$ have compact supports. Then eq. (\ref{3.2.5}) reduces to
the form%
\begin{eqnarray}
&&J(m_{2|0},P_{2}^{(0)})=0,\;P_{2}^{(0)}=\tilde{C}%
_{2}^{(0)}+c_{4|1}C_{1|47|1},\;C_{1|47|k}=\left. M_{47}\right| _{\omega
\rightarrow \omega _{k}},  \label{3.2.6} \\
&&C_{1|47}(z|f,g)=\omega (f_{0})\Delta g(z)-\Delta f(z)\omega (g_{0}).
\notag
\end{eqnarray}

The general solution of eq. (\ref{3.2.6}) is ($\varepsilon _{P_{2}^{(0)}}=1$)%
\begin{eqnarray}
&&\tilde{C}_{2}^{(0)}=c_{a|2}m_{2|a}+d_{1}^{\mathrm{ad}%
}M_{1}-c_{4|1}C_{1|47|1},\;a=4,5,6,  \label{3.2.7} \\
&&M_{1}(z|f)=M_{(0,1)}(x|f_{0})+\xi M_{(1,0)}(x|f_{1})  \notag \\
&&d_{1}^{\mathrm{ad}%
}M_{1}(z|f,g)=[M_{1}(z|g),f(z)]-[M_{1}(z|f),g(z)]-M_{1}(z|[f,g]),  \notag
\end{eqnarray}

L.h.s. of eq. (\ref{3.2.7}) can be extended to the case $f,g,h\in DE$. We
will study the possibility of the similar extension of r.h.s. of eq. (\ref%
{3.2.7}). Rewrite eq. (\ref{3.2.7}) in the form%
\begin{eqnarray}
&&\!\!\!\!\!\!\!\!\!\!\!\!\!\!\!
\tilde{C}_{2}^{(0)}(z|f,g)-c_{a|2}m_{2|a}(z|f,g)+M_{1}(z|[f,g])+c_{4|1}%
\omega ([f,g])\equiv \tilde{P}_{2}^{(0)}(z|f,g)=  \notag \\
&&\!\!\!\!\!\!\!\!\!\!\!\!\!\!\!
\,=[M_{1}(z|g),f(z)]-[M_{1}(z|f),g(z)]-c_{4|1}[\omega _{1}(f_{0})\Delta
g(z)+(-1)^{\varepsilon (f)}\Delta f(z)\omega _{1}(g_{0})].  \label{3.2.8}
\end{eqnarray}%
L.h.s. of eq. (\ref{3.2.8}) can be as before extended to the case $f,g,h\in
DE$. Consider eq. (\ref{3.2.8}) for the functions $\hat{f}_{0}$, $\hat{g}%
_{1} $. We have%
\begin{equation*}
\tilde{P}_{2}^{(0)}(z|\hat{f}_{0},\hat{g}_{1})-M_{(1,0)}(x|g_{1})f_{0}^{%
\prime }(x)=M_{(0,1)}^{\prime }(x|f_{0})g_{1}(x)-c_{4|1}\omega
_{1}(f_{0})g_{1}^{\prime }(x).
\end{equation*}

Let $g_{1}(x)$ is fixed and is equal to $1$ in a neighborhood $U$ of a point
$x_{0}$. For $x\in U$, we have%
\begin{equation*}
M_{(0,1)}^{\prime }(x|f_{0})=\tilde{P}_{2}^{(0)}(z|\hat{f}_{0},\hat{g}%
_{1})-M_{(1,0)}(x|g_{1})f_{0}^{\prime }(x)
\end{equation*}%
and we can consider this formula as an extension of $M_{(0,1)}^{\prime
}(x|f_{0})\equiv \tilde{M}_{(0,1)}(x|f_{0})$ to $f_{0}\in E$. Let now
another $g_{1}(x)$ is fixed and is equal to $x$ in a neighborhood $V$ of a
point $x_{0}$. For $x\in V$, we have%
\begin{equation*}
c_{4|1}\omega _{1}(f_{0})=\tilde{M}%
_{(0,1)}(x|f_{0})x+M_{(1,0)}(x|g_{1})f_{0}^{\prime }(x)-\tilde{P}%
_{2}^{(0)}(z|\hat{f}_{0},\hat{g}_{1}).
\end{equation*}%
This relation means that the form $c_{4|1}\omega _{1}(f_{0})$ can be extended
to $f_{0}\in E$, what is possible only if%
\begin{equation*}
c_{4|1}\omega _{1}(f_{0})=0.
\end{equation*}

Thus we obtained%
\begin{equation*}
c_{4|1}c_{5|1}=c_{4|1}c_{6|1}=c_{4|1}\omega _{1}=0.
\end{equation*}

\subappen{A. $c_{4|1}\neq 0$}

First we consider the case $c_{4|1}\neq 0$. Then $c_{5|1}=c_{6|1}=\omega
_{1}=0$.

We have%
\begin{equation*}
C_{1}^{(0)}(\,z|f,g)=c_{4|1}m_{2|4}(\,z|f,g).
\end{equation*}

The function $C^{(0)}(z|f,g)$ can be represented in the form%
\begin{equation*}
C^{(0)}(z|f,g)=N_{c_{4|1}}(z|f,g)+\hbar ^{2}C_{2}^{(0)}(z|f,g)+O(\hbar
^{3}),\;\varepsilon _{C_{2}^{(0)}}=1.
\end{equation*}

\subsubappen{ Second order}

It follows from eq. (\ref{1.1a}) that%
\begin{eqnarray*}
J(m_{2|0},C_{2}^{(0)}) &=&0\;\Longrightarrow . \\
C_{2}^{(0)}(\,z|f,g) &=&c_{a|2}m_{2|a}(\,z|f,g)+m_{2|7|2}(\,z|f,g),\;a=4,5,6,
\end{eqnarray*}%
The function $C^{(0)}(z|f,g)$ can be represented in the form%
\begin{eqnarray*}
&&C^{(0)}(z|f,g)=N_{c_{4|[2]}}(z|f,g)+\hbar ^{2}c_{a|2}m_{2|a}(z|f,g)+\hbar
^{2}m_{2|7|2}(\,z|f,g)+ \\
&&\,+\hbar ^{3}C_{3}^{(0)}(z|f,g)+O(\hbar ^{4}),\;a=5,6,\;\varepsilon
_{C_{3}^{(0)}}=1,\;c_{b|[k]}=\sum_{l=1}^{k}\hbar ^{l-1}c_{b|l},\;b=4,5,6.
\end{eqnarray*}

\subsubappen{ Third order}

It follows from eq. (\ref{1.1a}) that%
\begin{eqnarray}
&&J(m_{2|0},\tilde{C}%
_{2}^{(0)})+c_{4|1}c_{a|2}J(m_{2|4},m_{2|a})+c_{4|1}J_{1}(m_{2|4},m_{2|7|2})=0,
\label{3.3.2.1} \\
&&a=5,6,\;\tilde{C}_{2}^{(0)}=C_{2}^{(0)}+c_{4|1}C_{47|2}  \notag
\end{eqnarray}

It follows from eq. (\ref{3.3.2.1}), similarly to solving eq. (\ref{3.2.0}),
that%
\begin{equation*}
c_{5|2}=c_{6|2}=\omega _{2}=0,
\end{equation*}

and so on.

Thus, we find: if $c_{4|1}\neq 0$, then the solution of eq. (\ref{1.1a}) has
the form%
\begin{equation*}
C^{(0)}(z|f,g)=\mathcal{N}_{c_{4|[\infty ]}}(z|f,g).
\end{equation*}

\subappen{B. $c_{4|1}=0$}

Now we consider the case when $c_{4|1}=0$ and at least one of $c_{5|1}$, $%
c_{6|1}$, $\omega _{1}$ does not equal to zero.

We have%
\begin{equation*}
C_{1}^{(0)}(\,z|f,g)=c_{a|1}m_{2|a}(\,z|f,g)+m_{2|7|1}(\,z|f,g),\;a=5,6.
\end{equation*}

The function $C^{(0)}(z|f,g)$ can be represented in the form%
\begin{equation*}
C^{(0)}=m_{2|0}+\hbar c_{a|1}m_{2|a}+\hbar m_{2|7|1}+\hbar
^{2}C_{2}^{(0)}+O(\hbar ^{3}),\;a=5,6,\;\varepsilon _{C_{2}^{(0)}}=1.
\end{equation*}

\subsubappen{ Second order}

It follows from eq. (\ref{1.1a}) that

\begin{eqnarray*}
&&J(m_{2|0},\tilde{C}_{2}^{(0)})=0,\;\tilde{C}%
_{2}^{(0)}=C_{2}^{(0)}+c_{6|1}C_{67|1}\;\Longrightarrow \\
&&C_{2}^{(0)}=c_{a|2}m_{2|a}+m_{2|7|2}-c_{6|1}C_{67|1},\;a=4,5,6.
\end{eqnarray*}

The function $C^{(0)}(z|f,g)$ can be represented in the form%
\begin{eqnarray*}
&&C^{(0)}(z|f,g)=\mathcal{N}_{c_{4|[2]}}(z|f,g)+\hbar
c_{a|[2]}m_{2|a}(z|f,g)+\hbar m_{2|7|[2]}(\,z|f,g)-\hbar ^{2}c_{6|1}C_{67|1}
\\
&&\,+\hbar ^{3}C_{3}^{(0)}(z|f,g)+O(\hbar ^{4}),\;a=5,6,\;\varepsilon
_{C_{3}^{(0)}}=1,\;m_{2|7|[k]}=\sum_{l=1}^{k}\hbar ^{l-1}m_{2|7|[l]}
\end{eqnarray*}%
\

\subsubappen{Third order}

It follows from eq. (\ref{1.1a}) (with equalities%
\begin{eqnarray}
&&J(m_{2|a},m_{2|b})=J(m_{2|5},m_{2|7})=J(m_{2|7})=J(m_{2|c},C_{67})=0,
\notag \\
&&a,b=5,6,c=5,6,7,  \label{3.4.2.1}
\end{eqnarray}%
taken into account) that%
\begin{eqnarray}
&&J(m_{2|0},\tilde{C}%
_{3}^{(0)})+c_{4|2}c_{b|1}J(m_{2|4},m_{2|b})+c_{4|2}J_{1}(m_{2|4},m_{2|7|1})=0,\;b=5,6,
\label{3.4.2.2} \\
&&\tilde{C}_{3}^{(0)}=C_{3}^{(0)}+c_{6|i}C_{67|i}+c_{4|2}C_{47|1},\;i,j=1,2,%
\;i+j=3.  \notag
\end{eqnarray}

It follows from eq. (\ref{3.4.2.2}), similarly to solving eq. (\ref{3.2.0}),
that%
\begin{eqnarray*}
c_{4|2}c_{5|1} &=&c_{4|2}c_{6|1}=c_{4|2}\omega _{1}=0\;\Longrightarrow
\;c_{4|2}=0\;\Longrightarrow \\
C_{3}^{(0)} &=&c_{a|3}m_{2|a}+m_{2|7|3}-c_{6|i}C_{67|i},\;a=4,5,6.
\end{eqnarray*}

The function $C^{(0)}(z|f,g)$ can be represented in the form%
\begin{eqnarray*}
&&C^{(0)}(z|f,g)=\mathcal{N}_{c_{4|[3]}}(z|f,g)+\hbar
c_{a|[3]}m_{2|a}(z|f,g)+\hbar m_{2|7|[3]}(\,z|f,g)- \\
&&\,-\hbar ^{2}c_{6|[2]}C_{67|[2]}+\hbar ^{4}C_{4}^{(0)}(z|f,g)+O(\hbar
^{5}),\;a=5,6,\;\varepsilon _{C_{3}^{(0)}}=1,
\end{eqnarray*}%
where%
\begin{equation*}
C_{67|[k]}=\sum_{l=1}^{k}\hbar ^{l-1}C_{67|l},\;[\sum_{l=0}^{K}\hbar
^{l}A_{l}]_{k}=\sum_{l=0}^{k}\hbar ^{l}A_{l}.
\end{equation*}

\subsubappen{Fourth order}

It follows from eq. (\ref{1.1a}) (with equalities (\ref{3.4.2.1}) and
\begin{equation}
J(C_{67})=0  \label{3.4.3.0}
\end{equation}
taken into account) that%
\begin{eqnarray}
&&J(m_{2|0},\tilde{C}%
_{4}^{(0)})+c_{4|3}c_{a|1}J(m_{2|4},m_{2|a})+c_{4|3}J_{1}(m_{2|4},m_{2|7|1})=0,\;a=5,6,
\label{3.4.3.1} \\
&&\tilde{C}_{4}^{(0)}=C_{4}^{(0)}+c_{6|3}C_{67|1}+c_{4|3}C_{47|1}.  \notag
\end{eqnarray}

It follows from eq. (\ref{3.4.3.1}), similarly to solving eq. (\ref{3.2.0}),
that%
\begin{eqnarray*}
c_{4|3}c_{5|1} &=&c_{4|3}c_{6|1}=c_{4|3}\omega _{1}=0\;\Longrightarrow
\;c_{4|3}=0\;\Longrightarrow \\
C_{4}^{(0)} &=&c_{a|4}m_{2|a}+m_{2|7|4}-c_{6|3}C_{67|1},\;a=4,5,6,
\end{eqnarray*}%
and so on.

Thus, we find: in the case when at least one of $c_{5|1}$, $c_{6|1}$, $%
\omega _{1}$ does not equal to zero, the solution of eq. (\ref{1.1a}) has
the form%
\begin{equation*}
C^{(0)}(z|f,g)=m_{2\backslash 0}(z|f,g)+\hbar c_{a|[\infty
]}m_{2|a}(z|f,g)+\hbar m_{2|7|[\infty ]}(\,z|f,g)-\hbar ^{2}c_{6|[\infty
]}C_{67|[\infty ]}.
\end{equation*}

\subappen{C. $c_{4|1}=c_{5|1}=c_{6|1}=\protect\omega _{1}=0$}

Representing $C^{(0)}$ in the form%
\begin{equation*}
C^{(0)}=m_{2\backslash 0}+\hbar ^{2}C_{2}^{(0)}+O(\hbar ^{3}),
\end{equation*}%
we obtain%
\begin{equation*}
J(C_{2}^{(0)};z|f,g,h)=0.
\end{equation*}%
Let in an $\hbar ^{k_{0}}$-order we have%
\begin{eqnarray*}
C^{(0)} &=&m_{2\backslash 0}+\hbar ^{^{k_{0}}}C_{^{k_{0}}}^{(0)}+O(\hbar
^{^{k_{0}}+1}), \\
J(m_{2|0},C_{^{k_{0}}}^{(0)}) &=&0\;\Longrightarrow
C_{^{k_{0}}}^{(0)}=c_{a|k_{0}}m_{2|a}(\,z|f,g)+m_{2|7|k_{0}}(\,z|f,g),
\end{eqnarray*}%
and at least one of the quantities $c_{4|k_{0}}$, $c_{5|k_{0}}$, $%
c_{6|k_{0}} $, $\omega _{k_{0}}$ does not equal to zero. For the folowing $%
C_{^{k_{0}+l}}^{(0)}$, $l=1,...,k_{0}-1$, we obtain as well%
\begin{equation*}
J(m_{2|0},C_{^{k_{0}+l}}^{(0)})=0\;\Longrightarrow
C_{^{k_{0}+l}}^{(0)}=c_{a|k_{0}+l}m_{2|a}(\,z|f,g)+m_{2|7|k_{0}+l}(\,z|f,g),
\end{equation*}%
such that we can represent $C^{(0)}$ in the form%
\begin{equation*}
C^{(0)}=m_{2\backslash
0}+c_{a|[2k_{0}-1]}m_{2|a}(\,z|f,g)+m_{2|7|[2k_{0}-1]}(\,z|f,g)+\hbar
^{^{2k_{0}}}C_{^{2k_{0}}}^{(0)}+O(\hbar ^{^{2k_{0}}+1}).
\end{equation*}%
It follows from eq. (\ref{1.1a}) the equation for $C_{^{2k_{0}}}^{(0)}$,%
\begin{eqnarray*}
&&J(m_{2|0},\tilde{C}%
_{2k_{0}}^{(0)})+c_{4|k_{0}}c_{a|k_{0}}J(m_{2|4},m_{2|a})+c_{4|k_{0}}J_{1}(m_{2|4},m_{2|7|k_{0}})=0,\;a=5,6,
\\
&&\tilde{C}%
_{4}^{(0)}=C_{4}^{(0)}+c_{6|k_{0}}C_{67|k_{0}}+c_{4|k_{0}}C_{47|k_{0}}\;%
\Longrightarrow \\
&&c_{4|k_{0}}c_{5|k_{0}}=c_{4|k_{0}}c_{6|k_{0}}=c_{4|k_{0}}\omega _{k_{0}}=0.
\end{eqnarray*}

Doing similarly to previous subsections, we obtain. In general case. there
are three types of the solutions of eq. (\ref{1.1a})

I.%
\begin{equation*}
C^{(0)}(z|f,g)=C^{(0)I}(z|f,g)=N_{c_{4|[\infty ]}}(z|f,g).
\end{equation*}

II.%
\begin{eqnarray}
&&C^{(0)}(z|f,g)=C^{(0)II}(z|f,g)=m_{2| 0}(z|f,g)+\hbar
c_{b|[\infty ]}m_{2|b}(z|f,g)+  \notag \\
&&+\hbar m_{2|7|[\infty ]}(\,z|f,g)-\hbar ^{2}c_{6|[\infty ]}C_{67|[\infty
]},\;b=5,6,  \label{3.5.1}
\end{eqnarray}%
where $c_{4|[\infty ]}$, $c_{a|[\infty ]}$, $m_{2|7|[\infty ]},$ and $%
C_{67|[\infty ]}$ are formal series in $\hbar $.

III.We single out especially the limit case%
\begin{equation*}
C^{(0)}(z|f,g)=C^{(0)III}(z|f,g)=m_{2\backslash 0}(z|f,g).
\end{equation*}

\appen{Solution of eq. (\protect\ref{1.1b})}

Let $C^{(1)}(z|f,g)=C_{1}^{(1)}(z|f,g)+O(\hbar )$, $C_{1}^{(1)}(z|f,g)=%
\theta _{\alpha |1}A_{\alpha |1}^{(1)}(z|f,g)$, $\varepsilon _{A_{\alpha
|1}^{(1)}}=0$, where $\theta _{\alpha |1}$ is a finite set of odd generators
of the Grassmann algebra of parameters. Sometimes, it is convenient to use
for $C_{1}^{(1)}(z|f,g)$ a representation%
\begin{eqnarray*}
&&C_{1}^{(1)}(z|f,g)=\int dvdu(-1)^{\varepsilon
(f)}[-a_{(000)|1}^{(1)}(x|y_{u},y_{v})+\xi \eta
_{u}a_{(110)|1}^{(1)}(x|y_{u},y_{v})- \\
&&-\xi \eta _{v}a_{(110)|1}^{(1)}(x|y_{v},y_{u})+\eta _{u}\eta
_{v}a_{(011)|1}^{(1)}(x|y_{u},y_{v})]f(u)g(v), \\
&&a_{(000)|1}^{(1)}(x|y_{v},y_{u})=-a_{(000)|1}^{(1)}(x|y_{u},y_{v}),%
\;a_{(011)|1}^{(1)}(x|y_{v},y_{u})=a_{(011)|1}^{(1)}(x|y_{u},y_{v}), \\
&&C_{1}^{(1)}(z|\hat{f}_{0},\hat{g}_{0})=a_{(011)|1}^{(1)}(x|f_{0},g_{0}),%
\;C_{1}^{(1)}(z|\hat{f}_{0},\hat{g}_{1})=\xi
a_{(110)|1}^{(1)}(x|f_{0},g_{1}), \\
&&C_{1}^{(1)}(z|\hat{f}_{1},\hat{g}_{0})=-\xi
a_{(110)|1}^{(1)}(x|g_{0},f_{1}),\;C_{1}^{(1)}(z|\hat{f}_{1},\hat{g}%
_{1})=a_{(000)|1}^{(1)}(x|f_{1},g_{1}), \\
&&a_{(?)|1}^{(1)}(x|f_{1},g_{1})=O(\theta )
\end{eqnarray*}

\subappen{First order}

It follows from eq. (\ref{1.1b}) that%
\begin{equation}
J(m_{2|0},C_{1}^{(1)};\,z|f,g,h)=0.  \label{4.1.1}
\end{equation}%
The general solution of eq \ (\ref{4.1.1}) for $f,g\in D_{1}$ was found in
Section \ref{cohom-res}:%
\begin{eqnarray*}
&&C_{1}^{(1)}(z|f,g)_{D}=m_{2}^{(1)}(z|f,g)+m_{2}^{(2)}(z|f,g),\;f,g\in
D_{1}, \\
&&m_{2}^{(1)}(z|f,g)=c_{a|1}m_{2|a}(z|f,g)-M_{1|1}(z|[f,g]),\;a=1,2,3, \\
&&m_{2}^{(2)}(z|f,g)=[M_{1|1}(z|f),g(z)]-(-1)^{(\varepsilon
(f)+1)(\varepsilon (g)+1)}[M_{1|1}(z|g),f(z)]\}, \\
&&M_{1|1}(z|\hat{f}_{0})=\xi \nu _{\alpha (11)|1}(x|f_{0}),\;M_{1|1}(z|\hat{f%
}_{1})=\nu _{(00)|1}(x|f_{1}),
\end{eqnarray*}%
$C_{1}^{(1)}(z|f,g)_{D}$ means a restriction $C_{1}^{(1)}(z|f,g)$ to the
space $f,g\in D_{1}$. The form $m_{2}^{(1)}(z|f,g)$ are well defined for $%
f,g\in ED$. Represent $C_{1}^{(1)}(z|f,g)$ in the form%
\begin{eqnarray}
&&C_{1}^{(1)}(z|f,g)=m_{2}^{(1)}(z|f,g)+M_{2|1}^{(2)}(z|f,g),\;f,g\in ED,
\notag \\
&&M_{2|1}^{(2)}(z|f,g)_{D}=m_{2}^{(2)}(z|f,g),\;f,g\in D_{1}.  \label{4.1.2}
\end{eqnarray}%
Represent $M_{2|1}^{(2)}(z|f,g)$ in the form%
\begin{eqnarray*}
&&M_{2|1}^{(2)}(z|f,g)=\int dvdu(-1)^{\varepsilon (f)}[-\mu
_{(000)|1}^{(2)}(x|y_{u},y_{v})+\xi \eta _{u}\mu _{(110)|1}(x|y_{u},y_{v})-
\\
&&-\xi \eta _{v}\mu _{(110)|1}(x|y_{v},y_{u})+\eta _{u}\eta _{v}\mu
_{(011)|1}(x|y_{u},y_{v})]f(u)g(v), \\
&&M_{2|1}^{(2)}(z|\hat{f}_{0},\hat{g}_{0})=\mu
_{(011)|1}^{(2)}(x|f_{0},g_{0}),\;M_{2|1}^{(2)}(z|\hat{f}_{1},\hat{g}%
_{1})=\mu _{(000)|1}^{(2)}(x|f_{1},g_{1}), \\
&&M_{2|1}^{(2)}(z|\hat{f}_{0},\hat{g}_{1})=\xi \mu
_{(110)|1}^{(2)}(x|f_{0},g_{1}),\;M_{2|1}^{(2)}(z|\hat{f}_{1},\hat{g}%
_{0})=-\xi \mu _{(110)|1}^{(2)}(x|g_{0},f_{1}).
\end{eqnarray*}

Consider eq. (\ref{4.1.2}) in terms of different components.

\subsubappen{$f=\hat{f}_{0}$, $g=\hat{g}_{0}$}

We have%
\begin{equation}
M_{2|1}^{(2)}(z|\hat{f}_{0},\hat{g}_{0})_{D}=\mu
_{(011)|1}^{(2)}(x|f_{0},g_{0})_{D}=\nu _{(11)|1}(x|f_{0})g_{0}^{\prime
}(x)_{D}+f_{0}^{\prime }(x)\nu _{i(11)|1}(x|g_{0})_{D}.  \label{4.1.1.1}
\end{equation}

Let $f_{0}(x)=x$ for $x\in U$ and fixed, $U$ is some vicinity. It follows
from eq. (\ref{4.1.1.1}) for $x\in U$:%
\begin{equation}
\nu _{(11)|1}(x|g_{0})_{D}=\mu _{(011)|1}^{(2)}(x|f_{0},g_{0})_{D}-\nu
_{(11)|1}(x|f_{0})g_{0}^{\prime }(x)_{D}.  \label{4.1.1.2}
\end{equation}%
Rel. (\ref{4.1.1.2}) means that the forms $\nu _{(11)|1}(x|g_{0})_{D}$ can
be extended to the case $g_{0}\in E$ such that we obtain that%
\begin{equation*}
M_{2|1}^{(2)}(z|\hat{f}_{0},\hat{g}_{0})=\mu
_{(011)|1}^{(2)}(x|f_{0},g_{0})=\nu _{(11)|1}(x|f_{0})g_{0}^{\prime
}(x)+f_{0}^{\prime }(x)\nu _{(11)|1}(x|g_{0}).
\end{equation*}

\paragraph{$f=\hat{f}_{1}$, $g=\hat{g}_{1}$}

We have%
\begin{eqnarray*}
&&M_{2|1}^{(2)}(z|\hat{f}_{1},\hat{g}_{1})=\mu
_{(000)|1}^{(2)}(x|f_{1},g_{1})=-\nu _{(00)|1}^{\prime }(x|f_{1})g_{1}(x)+ \\
&&\,+\nu _{(00)|1}^{\prime }(x|g_{1})f_{1}(x).
\end{eqnarray*}

\paragraph{$f=\hat{f}_{0}$, $g=\hat{g}_{1}$}

We have%
\begin{eqnarray*}
M_{2|1}^{(2)}(z|\hat{f}_{0},\hat{g}_{1})_{D} &=&\xi \mu
_{(110)|1}^{(2)}(x|f_{0},g_{1})_{D}=\xi \lbrack \nu
_{(11)|1}(x|f_{0})g_{1}^{\prime }(x)-\nu _{(11)|1}^{\prime
}(x|f_{0})g_{1}(x)]\;\Longrightarrow \\
M_{2|1}^{(2)}(z|\hat{f}_{0},\hat{g}_{1}) &=&\xi \mu
_{(110)|1}^{(2)}(x|f_{0},g_{1})=\xi \lbrack \nu
_{(11)|1}(x|f_{0})g_{1}^{\prime }(x)-\nu _{(11)|1}^{\prime
}(x|f_{0})g_{1}(x)].
\end{eqnarray*}

Thus we obtain that%
\begin{equation*}
C_{1}^{(1)}(z|f,g)=c_{a|1}m_{2|a}(z|f,g)+d_{1}^{\mathrm{as}%
}M_{1|1}(z|f,g),\;f,g\in ED_{1},\;a=1,2,3.
\end{equation*}

Note that the last summand in $C_{1}^{(1)}$ gives the contribution to the
function $C(z|f,g)$ that can be canceled by a similarity transformation with
$T(z|f)=f(z)-\hbar M_{1|1}(\,z|f)+O(\hbar ^{2}\theta _{1}^{2})$ which does
not change the function $C^{(0)}(z|f,g)$.

The function $C^{(1)}(z|f,g)$ can be represented in the form%
\begin{eqnarray*}
&&C^{(1)}(z|f,g)=\Theta _{a|1}m_{2|a}(z|f,g)+\hbar
^{1}C_{2}^{(1)}(z|f,g)+O(\hbar ^{2}),\;a=1,2,3, \\
&&C_{2}^{(1)}(z|f,g)=.\theta _{\alpha }A_{\alpha
|2}^{(1)}(z|f,g),\;\varepsilon _{A_{\alpha |2}^{(1)}}=0,\;\Theta
_{a|1}=\theta _{\alpha }c_{\alpha a|1}.
\end{eqnarray*}

\subappen{I. $C^{(0)}(z|f,g)=C^{(0)I}(z|f,g)$}

Let $c_{4|[\infty ]}=\hbar ^{k_{0}-1}c_{4|k_{0}}+O(\hbar ^{k_{0}})$, .$%
c_{4|k_{0}}\neq 0$ Then we have%
\begin{eqnarray*}
&&C^{(1)I}(z|f,g)=\Theta _{p|[k_{0}]}m_{2|p}(z|f,g)+\hbar
^{k_{0}}C_{k_{0}+1}^{(1)I}(z|f,g)+O(\hbar ^{k_{0}+1}), \\
&&C_{k_{0}+1}^{(1)I}=\theta _{\alpha }A_{\alpha k_{0}+1}^{(1)I},\;\Theta
_{p|[k_{0}]}=\sum_{k=1}^{k_{0}}\hbar ^{k-1}\Theta _{p|k},\;\Theta
_{p|k}=\theta _{\alpha }c_{\alpha p|k},\;p=1,2,3.
\end{eqnarray*}

\subsubappen{ ($k_{0}+1$)-th order}

It follows from eq. (\ref{1.1b}) that%
\begin{equation}
J(m_{2|0},C_{k_{0}+1}^{(1)I})+c_{4|k_{0}}\Theta
_{p|1}J(m_{2|4},m_{2|p})=0,\;p=1,2,3.  \label{4.2.1.1}
\end{equation}

Let $f=\hat{f}_{1}$, $g=\hat{g}_{1}$, $h=\hat{h}_{1}$

In this case, we have $J(m_{2|4},m_{2|3};z|\hat{f}_{1},\hat{g}_{1},\hat{h}%
_{1})\equiv 0$. It follows from eq.(\ref{4.2.1.1})%
\begin{eqnarray*}
&&J(m_{2|0},C_{k_{0}+1}^{(1)I};z|\hat{f}_{1},\hat{g}_{1},\hat{h}%
_{1})+c_{4|k_{0}}\Theta _{q|1}J(m_{2|4},m_{2|a};z|\hat{f}_{1},\hat{g}_{1},%
\hat{h}_{1})=0,\;q=1,2\;\Longrightarrow \\
&&a_{(000)|k_{0}+1}^{(1)I}(x|f_{1}g_{1}^{\prime }-f_{1}^{\prime
}g_{1},h_{1})-[\partial _{x}a_{(000)|k_{0}+1}^{(1)I}(x|f_{1},g_{1})]h_{1}(x)+%
\mathrm{cycle}(f_{1},g_{1},h_{1})= \\
&&\,=-.c_{4|k_{0}}\Theta _{q|1}J(m_{2|4},m_{2|q};z|\hat{f}_{1},\hat{g}_{1},%
\hat{h}_{1}).
\end{eqnarray*}

Let
\begin{equation*}
\lbrack x\cup \mathrm{supp}(h_{1})]\cap \lbrack \mathrm{supp}(f_{1})\cup
\mathrm{supp}(g_{1})]=x\cap \mathrm{supp}(h_{1})=\varnothing .
\end{equation*}%
We have
\begin{equation*}
\hat{a}_{(000)|k_{0}+1}^{(1)I}(x|f_{1}g_{1}^{\prime }-f_{1}^{\prime
}g_{1},h_{1})=0\;\Longrightarrow \;\hat{a}_{\alpha
(000)|k_{0}+1}^{(1)I}(x|f_{1},h_{1})=0\;\Longrightarrow
\end{equation*}%
\begin{eqnarray*}
&&a_{(000)|k_{0}+1}^{(1)I}(x|f_{1},h_{1})=a_{1}(x|f_{1},h_{1})+a_{2}(x|f_{1},h_{1}),
\\
&&a_{1}(x|f_{1},h_{1})=\sum_{q=0}^{Q}\{\partial
_{x}^{q}f_{1}(x)a_{1}^{q}(x|h_{1})-a_{1}^{q}(x|f_{1})\partial
_{x}^{q}h_{1}(x)\}, \\
&&a_{2}(x|f_{1},h_{1})=\sum_{l=0}^{L}a_{2}^{2l+1}(x|\{\partial
^{2l+1}f_{1}\}h_{1}-f_{1}\partial ^{2l+1}h_{1}).
\end{eqnarray*}

Let
\begin{equation*}
\lbrack x\cup \mathrm{supp}(h_{1})]\cap \lbrack \mathrm{supp}(f_{1})\cup
\mathrm{supp}(g_{1})]=\varnothing .
\end{equation*}%
We have
\begin{eqnarray*}
&&[\partial
_{x}a_{(000)|k_{0}+12}^{(1)I}(x|f_{1},g_{1})h_{1}(x)+a_{(000)|k_{0}+12}^{(1)I}(x|f_{1}^{\prime }g_{1}-f_{1}g_{1}^{\prime },h_{1})]=
\\
&&\,=c_{4|k_{0}}\left( \Theta _{1|1}\overline{f_{1}^{\prime \prime \prime
}g_{1}}+2\Theta _{2|1}\overline{\theta _{x}f_{1}^{\prime \prime \prime }g_{1}%
}\right) h_{1}^{\prime }(x).
\end{eqnarray*}%
or
\begin{eqnarray*}
&&\sum_{l=0}^{L}[\partial _{x}\hat{a}_{2}^{2l+1}(x|\{\partial
^{2l+1}f_{1}\}h_{1}-f_{1}\partial ^{2l+1}h_{1})]h_{1}(x)+\sum_{q=0}^{Q}\hat{a%
}_{1}^{q}(x|f_{1}g_{1}^{\prime }-f_{1}^{\prime }g_{1})\partial
_{x}^{q}h_{1}(x)= \\
&&\,=c_{4|k_{0}}\left( \Theta _{1|1}\overline{f_{1}^{\prime \prime \prime
}g_{1}}+2\Theta _{2|1}\overline{\theta _{x}f_{1}^{\prime \prime \prime }g_{1}%
}\right) h_{1}^{\prime }(x)\;\Longrightarrow
\end{eqnarray*}%
\begin{equation*}
\hat{a}_{1}^{1}(x|f_{1}g_{1}^{\prime }-f_{1}^{\prime
}g_{1})=c_{4|k_{0}}\left( \Theta _{1|1}\overline{f_{1}^{\prime \prime \prime
}g_{1}}+2\Theta _{2|1}\overline{\theta _{x}f_{1}^{\prime \prime \prime }g_{1}%
}\right) .
\end{equation*}

Let $f_{1}(y)=1$ for $y\in \mathrm{supp}g_{1}$ such that we obtain that%
\begin{equation*}
\hat{a}_{1}^{1}(x|g_{1}^{\prime })=0.
\end{equation*}%
Let $f_{1}(y)=y$ for $y\in \mathrm{supp}g_{1}$ such that we obtain that%
\begin{equation*}
\hat{a}_{1}^{1}(x|(yg_{1})^{\prime }-2g_{1})=0\;\Longrightarrow \;\hat{a}%
_{1}^{1}(x|g_{1})=0\;\Longrightarrow .
\end{equation*}%
\begin{eqnarray}
&&c_{4|k_{0}}\left( c_{1|1}\overline{f_{1}^{\prime \prime \prime }g_{1}}%
+2c_{2|1}\overline{\theta _{x}f_{1}^{\prime \prime \prime }g_{1}}\right)
=0\;\Longrightarrow  \notag \\
&&\,\Longrightarrow \;c_{4|k_{0}}\Theta _{1|1}=c_{4|k_{0}}\Theta
_{2|1}=0\;\Longrightarrow  \label{4.2.1.1a}
\end{eqnarray}%
\begin{equation*}
\Theta _{1|1}=\Theta _{2|1}=0
\end{equation*}

Eq. (\ref{4.2.1.1}) reduces to the form%
\begin{equation}
J(m_{2|0},C_{k_{0}+1}^{(1)I})=-c_{4|k_{0}}\Theta _{3|1}J(m_{2|4},m_{2|3})=%
\mathrm{loc}.  \label{4.2.1.2}
\end{equation}

In the case when the support of some function does not intersect with the
support of another function or some neighborhood of $x$, eq. (\ref{4.2.1.2})
reduces to%
\begin{equation}
J(m_{2|0},C_{k_{0}+1}^{(1)I})=0.  \label{4.2.1.3}
\end{equation}

For the case $f,g,h\in D$, eq. (\ref{4.2.1.3}) was solved in \cite{JMP} and \cite{apoi} and we have%
\begin{equation*}
C_{k_{0}+1}^{(1)I}=\Theta _{p|k_{0}+1}m_{2|p}+d_{1}^{\mathrm{as}%
}M_{1|k_{0}+1}+C_{|k_{0}+1\mathrm{loc}}^{(1)},\;p=1,2,3.
\end{equation*}

Eq. (\ref{4.2.1.2}) reduces to%
\begin{equation*}
J(m_{2|0},C_{k_{0}+1\mathrm{loc}}^{(1)I})=-c_{4|k_{0}1}\Theta
_{3|1}J(m_{2|4},m_{2|3}).
\end{equation*}%
Such equation was analized in \cite{apoi}. It follows from the results of \cite{apoi} that%
\begin{equation}
c_{4|k_{0}}\Theta _{3|1}=0\;\Longrightarrow \;\Theta _{3|1}=0,
\label{4.2.1.4}
\end{equation}%
and so on. We obtain finally $C^{(1)I}=0$ (more exactly, $C^{(1)I}$ can be
canceled by a similarity transformation).

\subappen{II. $C^{(0)}(z|f,g)=C^{(0)II}(z|f,g)$}

Let $c_{b|[\infty ]}=c_{b|1}+O(\hbar )$, $b=5,6$, $m_{2|7|[\infty
]}=m_{2|7|1}+O(\hbar )$, at least one of the quantities $c_{b|1}$, $%
m_{2|7|1} $ is not equal to zero.Then we have%
\begin{equation*}
C^{(1)II}(z|f,g)=\Theta _{p|1}m_{2|p}(z|f,g)+\hbar
C_{2}^{(1)II}(z|f,g)+O(\hbar ^{2}),\;p`=1,2,3
\end{equation*}

\subsubappen{Second order}

Using the relations%
\begin{eqnarray}
J(m_{2|1,}m_{2|d}) &=&J(m_{2|2,}m_{2|5})=0,\;d=5,6,7,  \label{4.3.1.0a} \\
J(m_{2|2},m_{2|6})
&=&J(m_{2|0},C_{26}),\;J(m_{2|2},m_{2|7})=J(m_{2|0},C_{27}),
\label{4.3.1.0b}
\end{eqnarray}%
we obtain from eq. (\ref{1.1b}) that%
\begin{eqnarray}
&&J(m_{2|0},\tilde{C}_{2}^{(1)II})+\Theta
_{3|1}[c_{b|1}J(m_{2|b},m_{2|3})+J(m_{2|7|1},m_{2|3})]=0,  \label{4.3.1.1} \\
&&\tilde{C}_{2}^{(1)II}=C_{2}^{(1)II}+\Theta
_{2|1}[c_{6|1}C_{26}+C_{27|1}],\;b=5,6,  \notag \\
&&\tilde{C}^{(1)II}(z|f,g)=\int dvdu(-1)^{\varepsilon
(f)}[-a_{(000)}^{(1)II}(x|y_{u},y_{v})+\xi \eta
_{u}a_{(110)}^{(1)II}(x|y_{u},y_{v})-  \notag \\
&&-\xi \eta _{v}a_{(110)}^{(1)II}(x|y_{v},y_{u})+\eta _{u}\eta
_{v}a_{(011)}^{(1)II}(x|y_{u},y_{v})]f(u)g(v).  \notag
\end{eqnarray}

Let $f=\hat{f}_{0}$, $g=\hat{g}_{0}$, $h=\hat{h}_{1}$. It follows from (\ref%
{4.3.1.1}) that%
\begin{eqnarray}
&&[\partial
_{x}a_{(011)2}^{(1)II}(x|f_{0},g_{0})]h_{1}(x)+[a_{(110)|2}^{(1)II}(x|f_{0},h_{1})g_{0}^{\prime }(x)-
\notag \\
&&\,-a_{(011)|2}^{(1)II}(x|f_{0}^{\prime }h_{1},g_{0})+(f_{0}\leftrightarrow
g_{0})]=\Theta _{3|1}[c_{b|1}J(m_{2|b},m_{2|3};z|\hat{f}_{0},\hat{g}_{0},%
\hat{h}_{1})+  \notag \\
&&\,+J(m_{2|7|1},m_{2|3};z|\hat{f}_{0},\hat{g}_{0},\hat{h}_{1})],\;b=5,6.
\label{4.3.1.2}
\end{eqnarray}

Let
\begin{equation*}
\lbrack x\cup \mathrm{supp}(h_{1})]\cap \lbrack \mathrm{supp}(f_{0})\cup
\mathrm{supp}(g_{0})]=\varnothing
\end{equation*}%
and $f_{0},g_{0}\in D$. It follows from (\ref{4.3.1.2}) that%
\begin{eqnarray*}
&&\widehat{a_{(011)|2}^{(1)II\prime }}(x|f_{0},g_{0})=0\;\Longrightarrow \\
&&a_{(011)|2}^{(1)II\prime }(x|f_{0},g_{0})=\partial
_{x}\sum_{q=0}^{Q}a_{q}(x|f_{0})\partial _{x}^{q}g_{0}(x)+b(x|f_{0})g_{0}(x)+
\\
&&\,+(f_{0}\leftrightarrow g_{0})\;\Longrightarrow
\end{eqnarray*}%
\begin{eqnarray*}
&&a_{(011)2}^{(1)II}(x|f_{0},g_{0})=\sum_{q=0}^{Q}a_{q}(x|f_{0})\partial
_{x}^{q}g_{0}(x)+ \\
&&\,+\int dy\theta
(x-y)b(y|f_{0})g_{0}(y)+d(f_{0},g_{0})+(f_{0}\leftrightarrow
g_{0}),\;f_{0},g_{0}\in D_{1}.
\end{eqnarray*}

Let
\begin{equation*}
\lbrack x\cup \mathrm{supp}(g_{0})]\cap \lbrack \mathrm{supp}(f_{0})\cup
\mathrm{supp}(h_{1})]=x\cap \mathrm{supp}(g_{0})=\varnothing ,
\end{equation*}%
$f_{0},g_{0}\in D$ and $f_{0}(y)=y$ for $y\in \mathrm{supp}h_{1}$. It
follows from (\ref{4.3.1.2}) that%
\begin{eqnarray*}
&&\widehat{\tilde{c}_{(011)2}^{(1)II}}(x|h_{1},g_{0})=0\;\Longrightarrow
\int dy\theta (x-y)b(y|h_{1})g_{0}(y)+\hat{d}(h_{1},g_{0})=0\;\Longrightarrow
\\
&&\hat{d}(h_{1},g_{0})=0\;\Longrightarrow b(y|h_{1})=0.
\end{eqnarray*}

Let
\begin{equation*}
\lbrack x\cup \mathrm{supp}(g_{0})]\cap \lbrack \mathrm{supp}(f_{0})\cup
\mathrm{supp}(h_{1})]=\varnothing
\end{equation*}%
and $f_{0},g_{0}\in D$. It follows from eq. (\ref{4.3.1.2}) that%
\begin{eqnarray*}
&&\widehat{a_{(110)|2}^{(1)II}}(x|f_{0},h_{1})g_{0}^{\prime
}(x)-\sum_{q=0}^{Q}\hat{a}_{q}(x|f_{0}^{\prime }h_{1})\partial
_{x}^{q}g_{0}(x)= \\
&&\,=\Theta _{3|1}[c_{5|1}\overline{f_{0}^{\prime }h_{1}^{\prime }}+c_{6|1}%
\overline{\theta _{x}f_{0}^{\prime }h_{1}^{\prime }}+\omega
_{1}(f_{0}^{\prime }h_{1})]g_{0}(x)\;\Longrightarrow .
\end{eqnarray*}%
\begin{equation*}
\hat{a}_{0}(x|f_{0}^{\prime }h_{1})=\Theta _{1}^{3}[c_{5|1}\overline{%
f_{0}^{\prime }h_{1}^{\prime }}+c_{6|1}\overline{\theta _{x}f_{0}^{\prime
}h_{1}^{\prime }}+\omega _{1}(f_{0}^{\prime }h_{1})]\;\Longrightarrow
\end{equation*}%
\begin{equation*}
\hat{a}_{0}(x|f_{0}^{\prime }h_{1})=\Theta _{3|1}[c_{5|1}\overline{%
f_{0}^{\prime }h_{1}^{\prime }}+\omega _{1}(f_{0}^{\prime }h_{1})],\;\Theta
_{3|1}c_{6|1}\overline{\theta _{x}f_{0}^{\prime }h_{1}^{\prime }}%
=0\;\Longrightarrow
\end{equation*}%
\begin{equation*}
\Theta _{3|1}c_{5|1}=\Theta _{3|1}c_{6|1}=0.
\end{equation*}

Eq. (\ref{4.3.1.1}) takes the form%
\begin{equation*}
J(m_{2|0},\tilde{C}_{2}^{(1)II})+\Theta _{3|1}J(m_{2|7|1},m_{2|3})=0.
\end{equation*}

Let $f_{0},g_{0},h_{0}\in D$. The we have%
\begin{eqnarray}
&&J(m_{2|7|1},m_{2|3};z|f,g,h)=J(m_{2|0},C_{37|1};z|f,g,h)\;\Longrightarrow
\notag \\
&&J(m_{2|0},\tilde{C}_{2}^{(1)II}+\Theta _{3|1}C_{37|1})=0\;\Longrightarrow
\notag \\
&&\tilde{C}_{2}^{(1)II}+\Theta _{3|1}C_{37|1}=m_{2,2}^{(1)}+m_{2,2}^{(2)},
\label{4.3.1.3} \\
&&m_{2,2}^{(1)}(z|f,g)=\Theta
_{a|2}m_{2|a}(z|f,g)-M_{1|2}(z|[f,g]),\;a=1,2,3,  \notag \\
&&m_{2,2}^{(2)}=\{[M_{1|2}(z|f),g(z)]-(-1)^{(\varepsilon (f)+1)(\varepsilon
(g)+1)}[M_{1|2}(z|g),f(z)]\},  \notag \\
&&M_{1|2}(z|\hat{f}_{0})=\xi \nu _{(11)|2}(x|f_{0}),\;M_{1|2}(z|\hat{f}%
_{1})=\nu _{(00)|2}(x|f_{1}).  \notag
\end{eqnarray}

Rewrite eq. (\ref{4.3.1.3}) in the form%
\begin{eqnarray*}
P(z|f,g) &=&m_{2,2}^{(2)}(z|f,g)-\Theta
_{3|1}C_{37|1}(z|f,g),\;f_{0},g_{0}\in D, \\
P(z|f,g) &=&\tilde{C}_{2}^{(1)II}(z|f,g)-m_{2,2}^{(1)}(z|f,g),
\end{eqnarray*}%
the form $P(z|f,g)$ can be extended to $f_{0},g_{0}\in E$.

Let $f=\hat{f}_{0}$, $g=\hat{g}_{0}$ and $y<x$ for $y\in \mathrm{supp}g_{0}$
and .$y\in \mathrm{supp}g_{0}$ We have%
\begin{equation}
\Theta _{3|1}\omega _{1}(f_{0}g_{0})=P(z|\hat{f}_{0},\hat{g}_{0}).
\label{4.3.1.4}
\end{equation}%
It follows from eq. (\ref{4.3.1.4}) that the form $\Theta _{3|1}\omega
_{1}(f_{0}g_{0})$ can be extended to $f_{0},g_{0}\in E$ what is possible for
the case $\Theta _{3|1}\omega _{1}=0$ only. Thus we obtain that%
\begin{equation*}
\Theta _{3|1}c_{5|1}=\Theta _{3|1}c_{6|1}=\Theta _{3|1}\omega
_{1}=0\;\Longrightarrow \Theta _{3|1}=0\;\Longrightarrow
\end{equation*}%
\begin{eqnarray*}
&&J(m_{2|0},\tilde{C}_{2}^{(1)II})=0\;\Longrightarrow \\
&&C_{2}^{(1)II}=\Theta _{p|2}m_{2|p}-\Theta
_{2|1}[c_{6|1}C_{26}+C_{27|1}],\;p=1,2,3.
\end{eqnarray*}

Represent $C^{(1)II}$ in the form%
\begin{eqnarray*}
&&C^{(1)II}(z|f,g)=\Theta _{p|[2]}m_{2|p}(z|f,g)-\hbar \Theta
_{2|1}[c_{6|1}C_{26}+C_{27|1}]+ \\
&&\,+\hbar ^{2}C_{3}^{(1)II}(z|f,g)+O(\hbar ^{3}),\;p=1,2,3,\;\Theta _{3|1}=0
\end{eqnarray*}

\subsubappen{ Third order}

It follows from eq. (\ref{1.1b}) that%
\begin{eqnarray*}
&&J(m_{2|0},\tilde{C}_{3}^{(1)II})+\Theta
_{3|2}[c_{b|1}J(m_{2|b},m_{2|3})+J(m_{2|7|1},m_{2|3})]=0, \\
&&\tilde{C}_{3}^{(1)II}=C_{3}^{(1)II}+\hbar ^{-1}\left( \left. \Theta
_{2|[2]}[c_{6|[2]}C_{26}+C_{27|[2]}]\right| _{1}\right) ,\;b=5,6,
\end{eqnarray*}%
where we used relations (\ref{4.3.1.0a}), (\ref{4.3.1.0b}), and%
\begin{eqnarray}
&&J(m_{2|1},C_{67})=J(m_{2|b},C_{26})=J(m_{2|b},C_{27})=J(m_{2|7},C_{27})=
\notag \\
&&\,=J(m_{2|7})=0,\;J(m_{2|7},C_{26})+J(m_{2|2},C_{67})=0,\;b=5,6,
\label{4.3.2.1}
\end{eqnarray}%
and $\left. A(\hbar )\right| _{n}$ means the term of the $\hbar ^{n}$-order
of the Taylor series of $A(\hbar )$.

Using the results of previous subsubsec., we find%
\begin{eqnarray*}
&&\Theta _{3|2}=0, \\
&&C_{3}^{(1)II}=\Theta _{a|3}m_{2|a}-\hbar ^{-1}\left( \left. \Theta
_{2|[2]}[c_{6|[2]}C_{26}+C_{27|[2]}]\right| _{1}\right) ,\;a=1,2,3.
\end{eqnarray*}

Represent $C^{(1)II}$ in the form%
\begin{eqnarray*}
&&C^{(1)II}(z|f,g)=\Theta _{p|[3]}m_{2|p}(z|f,g)-\hbar \left( \left. \left[
\Theta _{2|[2]}(c_{6[2]}C_{26}+C_{27|[2]})\right] \right| _{[1]}\right) + \\
&&\,+\hbar ^{3}C_{4}^{(1)II}(z|f,g)+O(\hbar ^{4}),\;p=1,2,3,\;\Theta
_{3|[2]}=0,
\end{eqnarray*}%
where $\left. A(\hbar )\right| _{[n]}$ means the part of the Taylor series
of $A(\hbar )$ up to terms of $\hbar ^{n}$-order.

\subsubappen{ Fourth order}

It follows from eq. (\ref{1.1b}) that%
\begin{eqnarray}
&&J(m_{2|0},\tilde{C}_{4}^{(1)II})+\Theta
_{3|3}[c_{b|1}J(m_{2|b},m_{2|3})+J(m_{2|7|1},m_{2|3})]=0,  \label{4.3.3.1} \\
&&\tilde{C}_{4}^{(1)II}=C_{4}^{(1)II}+\hbar ^{-2}\left( \left. \Theta
_{\lbrack 3]}^{2}[c_{6|[3]}C_{26}+C_{27|[3]}]\right| _{2}\right) +\Theta
_{2|1}c_{6|1}^{2}C_{2667|1},\;b=5,6,  \notag
\end{eqnarray}%
where we used relations (\ref{4.3.1.0a}), (\ref{4.3.1.0b}), (\ref{4.3.2.1}),
and%
\begin{equation}
J(C_{26},C_{67})=J(m_{2|0},C_{2667}),\;J(C_{27},C_{67})=0.  \label{4.3.3.2}
\end{equation}

It follows from (\ref{4.3.3.1}) that $C^{(1)II}$ can be represented in the
form%
\begin{eqnarray*}
&&C^{(1)II}(z|f,g)=\Theta _{p|[4]}m_{2|p}(z|f,g)-\hbar \left( \left. \left[
\Theta _{2|[3]}(c_{6[3]}C_{26}+C_{27|[3]})\right] \right| _{[2]}\right) - \\
&&\,-\hbar ^{3}\Theta _{2|1}c_{6|1}^{2}C_{2667|1}+\hbar
^{4}C_{5}^{(1)II}(z|f,g)+O(\hbar ^{5}),\;p=1,2,3,\;\Theta _{3|[3]}=0,
\end{eqnarray*}

\subsubappen{Fifth order}

It follows from eq. (\ref{1.1b}) that%
\begin{eqnarray}
&&J(m_{2|0},\tilde{C}_{5}^{(1)II})+\Theta
_{3|4}[c_{b|1}J(m_{2|b},m_{2|3})+J(m_{2|7|1},m_{2|3})]=0,  \label{4.3.4.1} \\
&&\tilde{C}_{5}^{(1)II}=C_{5}^{(1)II}+\hbar ^{-3}\left( \left. \Theta
_{2|[4]}[c_{6|[4]}C_{26}+C_{27|[4]}]\right| _{3}\right) +  \notag \\
&&\,+\hbar ^{-1}\left( \left. \Theta
_{2|[4]}c_{6|[4]}^{2}C_{2667|[4]}\right| _{1}\right) ,\;b=5,6,  \notag
\end{eqnarray}%
where we used relations (\ref{4.3.1.0a}), (\ref{4.3.1.0b}), (\ref{4.3.2.1}),
(\ref{4.3.3.2}) and%
\begin{equation*}
J(m_{2|d},C_{2667})=J(C_{67},C_{2667})=0,\;d=5,6,7.
\end{equation*}

It follows from (\ref{4.3.4.1}) that $C^{(1)II}$ can be represented in the
form%
\begin{eqnarray*}
&&C^{(1)II}(z|f,g)=\Theta _{p|[5]}m_{2|p}(z|f,g)-\hbar \left( \left. \left[
\Theta _{2|[4]}(c_{6[4]}C_{26}+C_{27|[4]})\right] \right| _{[3]}\right) - \\
&&\,-\hbar ^{3}\left( \left. \Theta 2|_{[4]}c_{6|[4]}^{2}C_{2667|[4]}\right|
_{[1]}\right) +\hbar ^{5}C_{6}^{(1)II}(z|f,g)+O(\hbar
^{6}),\;p=1,2,3,\;\Theta _{3|[4]}=0,
\end{eqnarray*}%
and so on.

Finally, we obtain: general solution of eq. (\ref{1.1b}) for the case $%
C^{(0)}=C^{(0)II}$ \ and at least one of the quantities $c_{b|1}$, $b=5,6$, $%
m_{2|7|1}$ is not equal to zero is%
\begin{eqnarray}
&&C^{(1)II}(z|f,g)=\Theta _{q|[\infty ]}m_{2|q}(z|f,g)-\hbar \Theta
_{2|[\infty ]}(c_{6[\infty ]}C_{26}+C_{27|[\infty ]})-  \notag \\
&&\,-\hbar ^{3}\Theta _{2|[\infty ]}c_{6|[\infty ]}^{2}C_{2667|[\infty
]},\;q=1,2.  \label{4.3.4.2}
\end{eqnarray}%
It can be analogously provn that general solution of eq. (\ref{1.1b}) has as
well form (\ref{4.3.4.2}) in the case $c_{a|[\infty ]}=\hbar
^{k_{0}-1}c_{a|k_{0}}+O(\hbar ^{k_{0}})$, $m_{2|7|[\infty ]}=\hbar
^{k_{0}-1}m_{2|7|k_{0}}+O(\hbar ^{k_{0}})$, $k_{0}>1$, and at least one of
the quantities $c_{a|k_{0}}$, $m_{2|7|k_{0}}$ is not equal to zero.

\subappen{III. $C^{(0)}(z|f,g)=C^{(0)III}(z|f,g)$}

In this case, eq. (\ref{1.1b}) takes the form%
\begin{equation*}
J(m_{2|0},C^{(1)III})=0,
\end{equation*}%
such that we obtain%
\begin{equation*}
C^{(1)III}(z|f,g)=\Theta _{a|[\infty ]}m_{2|a}(z|f,g),\;a=1,2,3.
\end{equation*}

\appen{Extension to exact solution}

\subappen{Case $I$}

In this case, the form $C=C^{(0)I}=\mathcal{N}_{c_{4|[\infty ]}}$ is an
exact solution.

\subappen{Case $II$}

In this case, the form $C=C^{(0)II}+\hbar C^{(1)II}$ is an exact solution.
When proving, we used the relations%
\begin{eqnarray*}
&&J(m_{2|b},m_{2|b^{\prime
}})=J(m_{2|5},m_{2|7})=0,\;J(m_{2|6},m_{2|7})=J(m_{2|0},C_{67}), \\
&&J(m_{2|b},m_{2|1})=J(m_{2|5},m_{2|2})=0,%
\;J(m_{2|6},m_{2|2})=J(m_{2|0},C_{26}), \\
&&J(m_{2|b},C_{26})=J(m_{2|b},C_{27})=J(m_{2|b},C_{2667})=0, \\
&&J(m_{2|7})=J(m_{2|7},C_{67})=J(m_{2|7},m_{2|1})=0,%
\;J(m_{2|7},m_{2|2})=J(m_{2|0},C_{27}), \\
&&J(m_{2|7},C_{27})=J(m_{2|7},C_{2667})=0,\;J(C_{67})=J(C_{67},m_{2|1})=0 \\
&&J(m_{2|7},C_{26})+J(C_{67},m_{2|2})=0,%
\;J(C_{67},C_{26})=J(m_{2|0},C_{2667}), \\
&&J(C_{67},C_{27})=J(C_{67},C_{2667})=0, \\
&&J(m_{2|1},C_{26})=J(m_{2|1},C_{27})=J(m_{2|1},C_{2667})=0.
\end{eqnarray*}

\subappen{Case $III$}

We find a solution in the form%
\begin{eqnarray*}
&&C^{III}=m_{2|0}+\hbar \Theta _{p|[\infty ]}m_{2|p}+\hbar
^{2}C^{(2)III},\;C^{(2)III}=\gamma ^{i}C_{i}^{(2)III},\;p,i=1,2,3, \\
&&\gamma ^{1}=\Theta _{2|[\infty ]}\Theta _{3|[\infty ]},\;\gamma
^{2}=\Theta _{3|[\infty ]}\Theta _{1|[\infty ]},\;\gamma ^{3}=\Theta
_{1|[\infty ]}\Theta _{2|[\infty ]},\;\varepsilon _{C_{i}^{(2)III}}=1.
\end{eqnarray*}

We have%
\begin{equation*}
J(m_{2|0},C^{(2)III})+J(\Theta _{b}m_{2|q},\Theta _{3}m_{2|3})=0,\;q=1,2,
\end{equation*}%
or%
\begin{equation}
\gamma ^{i}J(m_{2|0},C_{i}^{(2)III})+\gamma ^{2}J(m_{2|1},m_{2|3})+\gamma
^{1}J(m_{2|3},m_{2|2})=0,  \label{5.3.1}
\end{equation}%
where we used relation $J(m_{2|1},m_{2|2})=0$.

Consider eq. (\ref{5.3.1}) for the functions $f=\hat{f}_{0}$, $g=\hat{g}_{1}$%
, $h=\hat{h}_{1}$,%
\begin{equation*}
\lbrack x\cup \mathrm{supp}(f_{0})]\cap \lbrack \mathrm{supp}(g_{1})\cup
\mathrm{supp}(h_{1})]=\varnothing ,
\end{equation*}%
and $f_{0}(x)=1$. We have%
\begin{equation}
\gamma ^{i}\hat{C}_{i}^{(2)III}(z|\eta (g_{1}h_{1}^{\prime }-g_{1}^{\prime
}h_{1}),\hat{f}_{0})+\gamma ^{2}\overline{g_{1}^{\prime \prime \prime }h_{1}}%
-2\gamma ^{1}\overline{\theta _{x}g_{1}^{\prime \prime \prime }h_{1}}=0
\label{5.3.2}
\end{equation}

Let $g(y)=1$ for $y\in \mathrm{supp}(h_{1})$. It follows from eq. (\ref%
{5.3.2}) that%
\begin{equation*}
\gamma ^{i}\hat{C}_{i}^{(2)III}(z|\eta h_{1}^{\prime },\hat{f}_{0})=0.
\end{equation*}

Let $g(y)=y$ for $y\in \mathrm{supp}(h_{1})$. It follows from eq. (\ref%
{5.3.2}) that%
\begin{eqnarray*}
&&\gamma ^{i}\hat{C}_{i}^{(2)III}(z|\eta (yh_{1}^{\prime }-h_{1}),\hat{f}%
_{0})=\gamma ^{i}\hat{C}_{i}^{(2)III}(z|\eta \lbrack (yh_{1})^{\prime
}-2h_{1}],\hat{f}_{0})=0\;\Longrightarrow \\
&&\gamma ^{i}\hat{C}_{i}^{(2)III}(z|\eta h_{1}),\hat{f}_{0})=0\;%
\Longrightarrow \;\gamma ^{2}\overline{g_{1}^{\prime \prime \prime }h_{1}}%
-2\gamma ^{1}\overline{\theta _{x}g_{1}^{\prime \prime \prime }h_{1}}%
=0\;\Longrightarrow \\
&&\gamma ^{1}=\gamma ^{2}=0,
\end{eqnarray*}%
or%
\begin{equation}
\Theta _{1|[\infty ]}\Theta _{3|[\infty ]}=\Theta _{2|[\infty ]}\Theta
_{3|[\infty ]}=0  \label{5.3.3}
\end{equation}

Eq. (\ref{5.3.3}) has two types of solutions.

$III_{1}$: $\Theta _{3|[\infty ]}=0$, $\Theta _{1|[\infty ]}$ and $\Theta
_{2|[\infty ]}$ are arbitrary.

$III_{2}$:%
\begin{equation*}
\Theta _{3|[\infty ]}\neq 0,\;\Theta _{q|[\infty ]}=c_{q|[\infty ]}\Theta
_{3|[\infty ]},\;q=1,2
\end{equation*}

Respectively, we have two types of exact solutions.

$III_{1}$:
\begin{equation*}
C^{III_{1}}=m_{2|0}+\hbar \Theta _{q|[\infty ]}m_{2|q},\;q=1,2.
\end{equation*}

$III_{2}$:%
\begin{equation*}
C^{III_{2}}=m_{2|0}+\hbar \Theta _{3|[\infty ]}(c_{q|[\infty
]}m_{2|q}+m_{2|3}),\;q=1,2.
\end{equation*}

\appen{Evaluation of general solution of eq. (\protect\ref{2.1})}\label{app6}

\subappen{Case $I$}

(2). In this case, we can represent the form $C$ as $C=C^{I}=C^{(0)I}+\hbar
^{2}C^{(2)I}+O(\theta ^{3})=\mathcal{N}_{c_{4|[\infty ]}}+\hbar
^{2}C^{(2)I}+O(\theta ^{3})$ where $C^{(2)I}$ has the properties $%
\varepsilon _{C^{(2)I}}=1$, $C^{(2)I}=O(\theta ^{2})$, and satisfies the
equation%
\begin{equation*}
J(\mathcal{N}_{c_{4|[\infty ]}},C^{(2)I})=0.
\end{equation*}

Let $c_{4|[\infty ]}=\hbar ^{k_{0}-1}c_{4|k_{0}}+O(\hbar ^{k_{0}})$, .$%
c_{4|k_{0}}\neq 0$ Then, in orders $k\leq k_{0},$ $C_{k}^{(2)I}$ satisfy the
equation%
\begin{equation*}
J(m_{2|0},C_{k}^{(2)I})=0
\end{equation*}%
and we have%
\begin{eqnarray*}
&&C^{(2)I}=\Theta _{a|[k_{0}]}^{(2)}m_{2|a}+m_{2|7}^{(2)|[k_{0}]}+\hbar
^{k_{0}}C_{k_{0}+1}^{(2)I}+O(\hbar ^{k_{0}+1}), \\
&&\Theta _{a|[k_{0}]}^{(2)}=\sum_{k=1}^{k_{0}}\hbar ^{k-1}\Theta
_{a|k}^{(2)},\;\Theta _{a|k}^{(2)}=O(\theta ^{2}),\;a=4,5,6. \\
&&m_{2|7}^{(2)||[k_{0}]}=\sum_{k=1}^{k_{0}}\hbar
^{k-1}m_{2|7}^{(2)|k},\;m_{2|7}^{(2)k}=O(\theta ^{2})
\end{eqnarray*}%
Represent $C^{I}$ in the form%
\begin{eqnarray*}
&&C^{I}\!
=\mathcal{N}_{c_{4|[\infty ]}^{[2]|[k_{0}]}}%
\!%
+\hbar ^{2}\Theta
_{b|[k_{0}]}^{(2)}m_{2|b}(z|f,g)+\hbar ^{2}m_{2|7}^{(2)|[k_{0}]}
\!%
+\hbar
^{k_{0}+2}C_{k_{0}+1}^{(2)I}(z|f,g)+O(\hbar ^{k_{0}+3})+O(\theta ^{3}), \\
&&c_{4|[\infty ]}^{[2]|[k_{0}]}=c_{4|[\infty ]}+\hbar ^{2}\Theta
_{4|[k_{0}]}^{(2)},\;b=5,6,\;J(N_{c_{4|[\infty ]}^{(2)|[k_{0}]}})=0.
\end{eqnarray*}

\subsubappen{ ($k_{0}+1$)-th order}

It follows from eq. (\ref{2.1}) that%
\begin{equation}
J(m_{2|0},C_{k_{0}+1}^{(2)I})+c_{4|k_{0}}\Theta
_{b|1}^{(2)}J(m_{2|4},m_{2|b})+c_{4|k_{0}}J(m_{2|4},m_{2|7}^{(2)|1})=0,%
\;b=5,6..  \label{6.1.1.1}
\end{equation}

As we seen above, see subsec 2 of sec 3, it follows from eq. (\ref{6.1.1.1})
that%
\begin{equation*}
\Theta _{b|1}^{(2)}=m_{2|7}^{(2)|1}=0,\;C_{k_{0}+1}^{(2)I}=\Theta
_{b|k_{0}+1}^{(2)}m_{2|b}+m_{2|7}^{(2)|k_{0}+1}.
\end{equation*}

Represent $C^{I}$ in the form%
\begin{eqnarray*}
&&C^{I}=\mathcal{N}_{c_{4|[\infty ]}^{(2)|[k_{0}+1]}}+\hbar ^{2}\Theta
_{b|[k_{0}+1]}^{(2)}m_{2|b}(z|f,g)+\hbar ^{2}m_{2|7}^{(2)|[k_{0}+1]}+\hbar
^{k_{0}+3}C_{k_{0}+2}^{(2)I}(z|f,g)+ \\
&&\,+O(\hbar ^{k_{0}+4})+O(\theta ^{3}),\;c_{4|[\infty
]}^{(2)|[k_{0}+1]}=c_{4|[\infty ]}+\hbar ^{2}\Theta
_{4|[k_{0}+1]}^{(2)},\;b=5,6, \\
&&J(N_{c_{4|[\infty ]}^{(2)|[k_{0}+1]}})=0.\;\Theta
_{b|1}^{(2)}=m_{2|7}^{(2)|1}=0,
\end{eqnarray*}%
and so on.

Finally we have%
\begin{equation*}
C^{I}=\mathcal{N}_{c_{4|[\infty ]}^{(2)|[\infty ]}}+\hbar
^{3}C^{(3)I}+O(\theta ^{4}),
\end{equation*}

(3). $C^{(3)I}$ satisfies the equation%
\begin{equation*}
J(\mathcal{N}_{c_{4|[\infty ]}},C^{(3)I})=0.
\end{equation*}

We have%
\begin{eqnarray*}
&&C^{(3)I}=\Theta _{p|[k_{0}]}^{(3)}m_{2|p}+\hbar
^{k_{0}}C_{k_{0}+1}^{(3)I}+O(\hbar ^{k_{0}+1}), \\
&&\Theta _{p|[k_{0}]}^{(3)}=\sum_{k=1}^{k_{0}}\hbar ^{k-1}\Theta
_{p|k}^{(3)},\;\Theta _{p|k}^{(3)}=O(\theta ^{3}),\;p=1.2.3.
\end{eqnarray*}%
In ($k_{0}+1$)-th order we obtain

\begin{equation}
J(m_{2|0},C_{k_{0}+1}^{(3)I})+c_{4|k_{0}}\Theta
_{p|1}^{(3)}J(m_{2|4},m_{2|p})=0,\;p=1,2,3.  \label{6.1.1.2}
\end{equation}

As we seen above, see subsubsec. 1 of subsec.2 of sec 4, it follows from eq.
(\ref{6.1.1.2}) that%
\begin{equation*}
\Theta _{p|1}^{(3)}=0,\;C_{k_{0}+1}^{(3)I}=\Theta _{p|k_{0}+1}^{(3)}m_{2|p}.
\end{equation*}

and so on. Thus we find $C^{(3)I}=0$ and

\begin{equation*}
C^{I}=\mathcal{N}_{c_{4|[\infty ]}^{(2)|[\infty ]}}+\hbar
^{4}C^{(4)I}+O(\theta ^{5}),\;J(\mathcal{N}_{c_{4|[\infty ]}},C^{(4)I})=0.
\end{equation*}

Finally we obtain%
\begin{eqnarray*}
&&C^{I}=\mathcal{N}_{c_{4|[\infty ]}^{[\infty ]}},\;c_{4|[\infty ]}^{[\infty
]}=\sum_{k,l=0}^{\infty }\hbar ^{k+2l-1}\theta ^{2l}c_{4|k}^{l}, \\
&&c_{4|k}^{0}=0,\;0\leq k\leq k_{0}-1,\;c_{4|k_{0}}^{0}=c_{4|k_{0}}\neq 0.
\end{eqnarray*}

\subappen{Case $II$}

(2). In this case, we can represent the form $C$ as $C=C^{II}=C^{(0)II}+%
\hbar C^{(1)II}+\hbar ^{2}C^{(2)II}+O(\theta ^{3})$, where $C^{(0)II}$ is
given by eq. (\ref{3.5.1}), $C^{(1)II}$ is given by eq. (\ref{4.3.4.2}), $%
J(C^{(0)II}+C^{(1)II})=0$, $C^{(2)II}$ has the properties $\varepsilon
_{C^{(2)II}}=1$, $C^{(2)II}=O(\theta ^{2})$, and satisfies the equation%
\begin{equation*}
J(C^{(0)II},C^{(2)II})=0.
\end{equation*}

Introduce notation%
\begin{eqnarray*}
&&C^{(0)II}+\hbar C^{(1)II}=\mathcal{N}^{II}(c_{b|[\infty ]},m_{2|7|[\infty
]},\Theta _{q|[\infty ]}), \\
&&b=5,6,\;q=1,2.
\end{eqnarray*}

Let $c_{b|[\infty ]}=\hbar ^{k_{0}-1}c_{b|k_{0}}+O(\hbar ^{k_{0}})$, $%
m_{2|7|[\infty ]}=\hbar ^{k_{0}-1}m_{2|7|k_{0}}+O(\hbar ^{k_{0}})$, $%
k_{0}\geq 1$, $b=5,6$, and at least one of the quantities $c_{b|k_{0}}$, $%
m_{2|7|k_{0}}$ is not equal to zero. Then, in orders $k\leq k_{0},$ $%
C_{k}^{(2)II}$ satisfy the equation%
\begin{equation*}
J(m_{2|0},C_{k}^{(2)II})=0
\end{equation*}%
and we have%
\begin{eqnarray*}
&&C^{(2)II}=\Theta _{a|[k_{0}]}^{(2)}m_{2|a}+m_{2|7}^{(2)|[k_{0}]}+\hbar
^{k_{0}}C_{k_{0}+1}^{(2)I}+O(\hbar ^{k_{0}+1}), \\
&&\Theta _{a|[k_{0}]}^{(2)}=\sum_{k=1}^{k_{0}}\hbar ^{k-1}\Theta
_{a|k}^{(2)},\;\Theta _{a|k}^{(2)}=O(\theta ^{2}),\;a=4,5,6. \\
&&m_{2|7}^{(2)||[k_{0}]}=\sum_{k=1}^{k_{0}}\hbar
^{k-1}m_{2|7}^{(2)|k},\;m_{2|7}^{(2)k}=O(\theta ^{2})
\end{eqnarray*}%
Represent $C^{II}$ in the form%
\begin{eqnarray*}
&&C^{II}\!=\!\mathcal{N}^{II}(c_{b|[\infty ]}^{[2]|[k_{0}]},m_{2|7|[\infty
]}^{[2]|[k_{0}]},\Theta _{q|[\infty ]})\!+\!\hbar ^{2}\Theta
_{4|[k_{0}]}^{(2)}m_{2|4}\!+\!\hbar ^{k_{0}\!+\!2}C_{k_{0}+1}^{(2)I}(z|f,g)\!+\!O(\hbar
^{k_{0}+3})\!+\!O(\theta ^{3}), \\
&&c_{b|[\infty ]}^{(2)|[k_{0}]}=c_{b|[\infty ]}+\hbar ^{2}\Theta
_{b|[k_{0}]}^{(2)|[k_{0}]},\;m_{2|7|[\infty ]}^{[2]|[k_{0}]}=m_{2|7|[\infty
]}+\hbar ^{2}m_{2|7}^{(2)|[k_{0}]},\;b=5,6.
\end{eqnarray*}

\subsubappen{ ($k_{0}+1$)-th order}

It follows from eq. (\ref{2.1}) that%
\begin{equation}
J(m_{2|0},C_{k_{0}+1}^{(2)II})+c_{b|k_{0}}\Theta
_{4|1}^{(2)}J(m_{2|b},m_{2|4})+\Theta
_{4|1}^{(2)}J(m_{2|7|k_{0}},m_{2|4})=0,\;b=5,6..  \label{6.2.1.1}
\end{equation}

As we seen above, see subsubsec. 2 of subsec. 4 of sec. 3, it follows from
eq. (\ref{6.2.1.1}) that%
\begin{equation*}
\Theta _{4|1}^{(2)}=0,\;C_{k_{0}+1}^{(2)II}=\Theta
_{a|k_{0}+1}^{(2)}m_{2|a}+m_{2|7}^{(2)|k_{0}+1}.
\end{equation*}

Represent $C^{II}$ in the form%
\begin{eqnarray*}
&&C^{II}=\mathcal{N}^{II}(c_{b|[\infty ]}^{[2]|[k_{0}+1]},m_{2|7|[\infty
]}^{[2]|[k_{0}+1]},\Theta _{q|[\infty ]})+\hbar ^{2}\Theta
_{4|[k_{0}+1]}^{(2)}m_{2|4}+ \\
&&+\hbar ^{k_{0}+3}C_{k_{0}+2}^{(2)II}(z|f,g)+O(\hbar ^{k_{0}+4})+O(\theta
^{3}), \\
&&c_{b|[\infty ]}^{(2)|[k_{0}+1]}=c_{b|[\infty ]}+\hbar ^{2}\Theta
_{b|[k_{0}]}^{(2)|[k_{0}+1]}, \\
&&m_{2|7|[\infty ]}^{[2]|[k_{0}+1]}=m_{2|7|[\infty ]}+\hbar
^{2}m_{2|7}^{(2)|[k_{0}+1]},\;b=5,6,\;\Theta _{4|1}^{(2)}=0.
\end{eqnarray*}%
and so on.

Finally we have%
\begin{equation*}
C^{II}=\mathcal{N}^{II}(c_{b|[\infty ]}^{[2]|[\infty ]},m_{2|7|[\infty
]}^{[2]|[\infty ]},\Theta _{q|[\infty ]})+\hbar ^{3}C^{(3)I}+O(\theta ^{4}),
\end{equation*}

(3). $C^{(3)II}$ satisfies the equation%
\begin{equation*}
J(C^{(0)II},C^{(3)II})=0.
\end{equation*}

We have%
\begin{eqnarray*}
&&C^{(3)I}=\Theta _{p|[k_{0}]}^{(3)}m_{2|p}+\hbar
^{k_{0}}C_{k_{0}+1}^{(3)I}+O(\hbar ^{k_{0}+1}), \\
&&\Theta _{p|[k_{0}]}^{(3)}=\sum_{k=1}^{k_{0}}\hbar ^{k-1}\Theta
_{p|k}^{(3)},\;\Theta _{p|k}^{(3)}=O(\theta ^{3}),\;p=1.2.3.
\end{eqnarray*}

Represent $C^{II}$ in the form%
\begin{eqnarray*}
&&C^{II}=\mathcal{N}^{II}(c_{b|[\infty ]}^{[2]|[\infty ]},m_{2|7|[\infty
]}^{[2]|[\infty ]},\Theta _{q|[\infty ]}^{(3)|[k_{0}]})+\Theta
_{3|[k_{0}]}^{(3)}m_{2|3}+ \\
&&+\hbar ^{k_{0}+3}C_{k_{0}+1}^{(3)II}+O(\hbar ^{k_{0}+4})+O(\theta ^{4}), \\
&&\Theta _{q|[\infty ]}^{(3)|[k_{0}]}=\Theta _{q|[\infty ]}+\Theta
_{q|[k_{0}]}^{(3)},\;q=1,2.
\end{eqnarray*}%
In ($k_{0}+1$)-th order we obtain

\begin{eqnarray}
&&J(m_{2|0},C_{k_{0}+1}^{(3)II})+c_{b|k_{0}}\Theta
_{3|1}^{(3)}J(m_{2|b},m_{2|3})+\Theta _{3|1}^{(3)}J(m_{2|7|k_{0}},m_{2|3})=
\notag \\
&&\,=0,\;5,6..  \label{6.2.1.2}
\end{eqnarray}

As we seen above, see subsubsec. 1 of subsec. 3 of sec 4, it follows from
eq. (\ref{6.2.1.2}) that%
\begin{eqnarray*}
&&\Theta _{3|1}^{(3)}=0,\;C_{k_{0}+1}^{(3)II}=\Theta
_{p|k_{0}+1}^{(3)}m_{2|p}\;\Longrightarrow \\
&&C^{II}=\mathcal{N}^{II}(c_{b|[\infty ]}^{[2]|[\infty ]},m_{2|7|[\infty
]}^{[2]|[\infty ]},\Theta _{q|[\infty ]}^{(3)|[k_{0}+1]})+\Theta
_{3|[k_{0}+1]}^{(3)}m_{2|3}+ \\
&&\,+\hbar ^{k_{0}+4}C_{k_{0}+2}^{(3)II}+O(\hbar ^{k_{0}+5})+O(\theta
^{4}),\;\Theta _{3|1}^{(3)}=0,
\end{eqnarray*}%
and so on, such that we obtain.%
\begin{eqnarray*}
&&C^{II}=\mathcal{N}^{II}(c_{b|[\infty ]}^{[2]|[\infty ]},m_{2|7|[\infty
]}^{[2]|[\infty ]},\Theta _{q|[\infty ]}^{(3)|[\infty ]})+\hbar
^{4}C^{(4)II}+O(\theta ^{5}), \\
&&J(C^{(0)II},C^{(4)II})=0
\end{eqnarray*}

Finally we find%
\begin{eqnarray*}
&&C^{II}=\mathcal{N}^{II}(c_{b|[\infty ]}^{[\infty ]},m_{2|7|[\infty
]}^{[\infty ]},\Theta _{q|[\infty ]}^{[\infty ]}), \\
&&c_{b|[\infty ]}^{[\infty ]}=\sum_{k=1}^{\infty }\sum_{l=0}^{\infty }\hbar
^{k+2l-1}\theta ^{2l}c_{b|k}^{l},\;c_{b|k}^{0}=0,\;0\leq k\leq k_{0}-1, \\
&&m_{2|7|[\infty ]}^{[\infty ]}=\sum_{k=1}^{\infty }\sum_{l=0}^{\infty
}\hbar ^{k+2l-1}\theta ^{2l}m_{2|7|k}^{l},\;m_{2|7|k}^{0}=0,\;0\leq k\leq
k_{0}-1, \\
&&\Theta _{q|[\infty ]}^{(2)[\infty ]}=\sum_{k=1}^{\infty
}\sum_{l=0}^{\infty }\hbar ^{k+2l-1}\theta
^{2l+1}c_{q|k}^{l},\;q=1,2,\;b=5,6,
\end{eqnarray*}%
where, at least, one of the numbers $c_{b|k_{0}}^{0}$, $m_{2|7|k_{0}}^{0}$
is not equal to zero.

\subappen{Case $III$, only three Grassmann generators $\protect\theta %
_{p}, $ $p=1,2,3$}

In this case, we have%
\begin{eqnarray*}
C &=&C^{III}=m_{2|0}+\theta _{p}c_{p|p^{\prime }}m_{2|p^{\prime }}+\gamma
^{i}C_{i}^{III(2)}+\beta C^{III(3)}, \\
\gamma ^{i} &=&\frac{1}{2}\varepsilon ^{ijk}\theta _{j}\theta _{k},\;\beta
=\theta _{1}\theta _{2}\theta _{3},\;p,p^{\prime },i,j,k=1,2,3.
\end{eqnarray*}

It follows from eq. (\ref{2.1}) that%
\begin{eqnarray}
&&J(m_{2|0},C_{i}^{III(2)})+\varepsilon
^{ijk}c_{j|q}c_{k|3}J(m_{2|q},m_{2|3})=0,  \label{6.3.4.1a} \\
&&J(m_{2|0},C^{III(3)})+c_{i|p}J(m_{2|p},C_{i}^{III(2)})=0,  \label{6.3.4.1b}
\end{eqnarray}%
where we used the equalities $\theta _{i}\theta _{j}=\varepsilon
^{ijk}\gamma ^{k}$, $\theta _{i}\gamma ^{k}=\delta _{i}^{k}\beta $.

It follows from eq. (\ref{6.3.4.1a}) (see subsec. 3 of sec. 5)%
\begin{eqnarray*}
&&\varepsilon ^{ijk}c_{j|q}c_{k|3}=0, \\
&&J(m_{2|0},C_{i}^{III(2)})\;\Longrightarrow
\;C_{i}^{III(2)}=d_{i|a}m_{2|a}+m_{2|7,i},\;a=4,5,6.
\end{eqnarray*}

Eq. (\ref{6.3.4.1b}) reduces to the form%
\begin{eqnarray*}
&&J(m_{2|0},\tilde{C}^{III(3)})+c_{i|q}d_{i|4}J(m_{2|q},m_{2|4})+ \\
&&\;+c_{i|3}J(m_{2|3},d_{i|a}m_{2|a}+m_{2|7,i})=0,\;q=1,2,\;a=4,5,6, \\
&&\tilde{C}^{III(3)}=C^{III(3)}+c_{i|2}(d_{i|6}C_{26}+C_{27,i}).
\end{eqnarray*}

Let $f=\hat{f}_{1}$, $g=\hat{g}_{1}$, $h=\hat{h}_{1}$. In this case $%
J(m_{2|3},d_{i|a}m_{2|a}+m_{2|7,i})=0$ and we obtain that%
\begin{equation*}
J(m_{2|0},\tilde{C}^{III(3)})+c_{i|q}d_{i|4}J(m_{2|q},m_{2|4})=0,
\end{equation*}%
and we find that (see subsubsec. 1 of subsec. 2 of sec.4, eq. (\ref{4.2.1.1a}%
))%
\begin{eqnarray*}
&&c_{i|q}d_{i|4}=0, \\
&&J(m_{2|0},\tilde{C}^{III(3)})+c_{i|3}J(m_{2|3},d_{i|a}m_{2|a}+m_{2|7,i})=0.
\end{eqnarray*}

Let $f=\hat{f}_{0}$, $g=\hat{g}_{0}$, $h=\hat{h}_{1}$ and their domains have
the properties $[x\cup \mathrm{supp}(h_{1})]\cap \lbrack \mathrm{supp}%
(f_{0})\cup \mathrm{supp}(g_{0})]=\varnothing $, or $[x\cup \mathrm{supp}%
(g_{0})]\cap \lbrack \mathrm{supp}(f_{0})\cup \mathrm{supp}(h_{1})]=x\cap
\mathrm{supp}(g_{0})=\varnothing $, or $x\cup \mathrm{supp}(g_{0})]\cap
\lbrack \mathrm{supp}(f_{0})\cup \mathrm{supp}(h_{1})]=\varnothing $. In
this case $J(m_{2|3},m_{2|4})=0$ and we obtain that%
\begin{equation}
J(m_{2|0},\tilde{C}^{III(3)})+c_{i|3}J(m_{2|3},d_{i|b}m_{2|b}+m_{2|7,i})=0,%
\;b=5,6.  \label{6.3.4.2}
\end{equation}

It follows from (\ref{6.3.4.2}) (see subsubsec. 1 of subsec. 3 of sec.4) that%
\begin{eqnarray}
&&c_{i|3}d_{i|b}=c_{i|3}\omega _{i}=0,  \notag \\
&&J(m_{2|0},\tilde{C}^{III(3)})+c_{i|3}d_{i|4}J(m_{2|3},m_{2|4})=0.
\label{6.3.4.3}
\end{eqnarray}

It follows from (\ref{6.3.4.3}) (see subsubsec. 1 of subsec. 2 of sec.4, eq.
(\ref{4.2.1.4})) that%
\begin{eqnarray*}
&&c_{i|3}d_{i|4}=0, \\
&&J(m_{2|0},\tilde{C}^{III(3)})=0\;\Longrightarrow \;\tilde{C}%
^{III(3)}=e_{p}m_{2|p^{\prime }}.
\end{eqnarray*}%
Thus, we obtain that the general solution of the type of $C^{III}$ with only
three generators $\theta _{i}$ has the form%
\begin{eqnarray*}
&&C^{III}=m_{2|0}+\hbar \theta _{p}c_{p|p^{\prime }}m_{2|p^{\prime }}+\hbar
^{2}\gamma ^{i}(d_{i|a}m_{2|a}+m_{2|7,i})+ \\
&&+\hbar ^{3}\beta \lbrack e_{p}m_{2|p}-c_{i|2}(d_{i|6}C_{26}+C_{27,i})], \\
&&p,p^{\prime }=1,2,3,\;q=1,2,\;a=4,5,6,
\end{eqnarray*}%
where $c_{p|p^{\prime }}$, $d_{i|a}$, $e_{p}$, and $\omega _{i}$ are
arbitrary series in $\hbar $ satisfying the conditions%
\begin{equation*}
\varepsilon ^{ijk}c_{j|q}c_{k|3}=c_{i|q}d_{i|4}=c_{i|3}d_{i|a}=c_{i|3}\omega
_{i}=0.
\end{equation*}




\begin{thebibliography}{99}




\bibitem{apoi} {\it S.~E.~Konstein, and I.~V.~Tyutin,}
The deformations of antibracket with even and odd deformation parameters,
arXiv: 1011.5807 [hep-th]

\bibitem{Schei97} {\it M.~Scheunert and R.~B.~Zhang,} J.Math.Phys., {\bf 39},
5024--5061 (1998); q-alg/9701037.

\bibitem{SKT1} {\it S.~E.~Konstein, A.~G.~Smirnov and I.~V.~Tyutin,}
Cohomologies of the Poisson superalgebra,
Teor.~Mat.~Fiz., {\bf 143},625 (2005);
hep-th/0312109.

\bibitem{JMP} {\it S.~E.~Konstein, and I.~V.~Tyutin,}
Deformations and central extensions of the antibracket Superalgebra,
Journal of Mathematical Physics, 49, 072103 (2008).

\bibitem{Leites} {\it D.~A.~Leites and I.~M.~Shchepochkina,}
How to quantize the antibracket, \\Theor.~Math.~Phys., {\bf
126}, 281--306 (2001).

\bibitem{leit} {\it D.Leites}, Clifford algebra as a superalgebra and quantization,
Theor.~Mat.~Fiz., {\bf 58} (1984), no. 2, 229--232.

\bibitem{Sche1} {\it M.~Sheunert}, Generalized Lie Algebras,
J.~Math.~Phys. {\bf 20} (1979),  no.4, 712-720.

\bibitem{Sosrus1} {\it D.Leites} (ed.),
"Seminar on Supersymmetry", vol. 1 "Algebra and Calculus"\\
(J.Bernstein, D.Leites, V.Shander, V.Molotkov), MCCME, Moscow, 2011 (in Russian).


\end{thebibliography}
\end{document}